\documentclass[12pt]{article}
\usepackage{kpfonts}

\usepackage[T1]{fontenc}
\usepackage{sfmath}
\usepackage{helvet}
\usepackage{amsmath}
\usepackage{amsthm}
\usepackage{nicefrac}
\usepackage{graphicx}
\usepackage{caption}
\usepackage{subcaption}
\usepackage[usenames, dvipsnames, svgnames, table]{xcolor} 
\usepackage[pagebackref=true]{hyperref}
\hypersetup{
	pdfborder={0 0 0.5 [2 1]},
	colorlinks=true,
	linkcolor=blue,
	citecolor=Maroon,
	urlcolor=blue,
	linkbordercolor=blue,
	citebordercolor=PineGreen,
	urlbordercolor=blue
}

\usepackage[capitalise,noabbrev]{cleveref}
\usepackage{amssymb}
\usepackage{geometry}
\usepackage{soul}
\usepackage{textcomp}
\newtheorem{example}{Example}
\newtheorem{lemma}{Lemma}

\newtheorem{claim}{Claim}
\newtheorem{proposition}{Proposition}
\newtheorem{corollary}{Corollary}
\newtheorem{theorem}{Theorem}
\renewcommand{\k}{\kappa}
\renewcommand{\i}{\iota}
\usepackage{xspace}
\newcommand{\arrng}{capacity profile\xspace}
\newcommand{\Arrng}{Capacity profile\xspace}
\newcommand{\arrngs}{capacity profiles\xspace}
\newcommand{\Arrngs}{Capacity Profiles\xspace}
\newcommand{\assignment}{matching\xspace}

\usepackage{comment}
\usepackage{verbatim}
\usepackage{appendix}
\usepackage{paralist}
\usepackage{natbib}

\makeatletter
\long\def\@makecaption#1#2{%
  \vskip\abovecaptionskip
  \sbox\@tempboxa{\textbf{#1}: #2}%
  \ifdim \wd\@tempboxa >\hsize
    \textbf{#1:} #2\par
  \else
    \global \@minipagefalse
    \hb@xt@\hsize{\box\@tempboxa\hfil}%
  \fi
  \vskip\belowcaptionskip}
\makeatother

\newif\ifimages
\imagestrue 

\definecolor{darkviolet}{rgb}{0.58, 0.0, 0.83} %

\DeclareMathOperator{\ida}{\mathbf{interview-DA}}
\DeclareMathOperator{\mda}{\mathbf{match-DA}}

\usepackage{array}
\newcolumntype{C}{>{$}c<{$}}

\title{Interview Hoarding\footnote{
    We thank Anna Sorensen and Alkas Baybas for raising the question
    that sparked this paper. We also thank
    Alex Chan,
    Adrienne Quirouet,
    Assaf Romm,
    Al Roth, Erling Skancke, Colin Sullivan, 
    William Thomson,
    and seminar audiences at
    North Carolina State University,
    Stanford,
    University of Arizona,
    and University of Lausanne
    for
    helpful comments and discussions. We are grateful for the 
    feedback from Federico Echenique and two anonymous referees
    that helped us improve the paper substantially.
  }}
\author{
\begin{tabular}{c}
Vikram Manjunath\\
{\footnotesize University of Ottawa}\\
{\footnotesize vikramma@gmail.com}
\end{tabular}
\hspace{1in}
\begin{tabular}{c}
  Thayer  Morrill\\
{\footnotesize North Carolina State University}\\
{\footnotesize thayer\_morrill@ncsu.edu}
\end{tabular}
}

\begin{document}

\maketitle
\begin{abstract}
  Many centralized matching markets are preceded by interviews between
   participants.  We study the impact on the final match
  of an increase in the number of interviews for one side of the market. Our
  motivation is the match between residents and hospitals
  where, due 
to the COVID-19 pandemic, interviews for the 2020-21
  season of the National Residency Matching Program  were switched to a virtual format. This 
  drastically reduced the cost to applicants of accepting interview
  invitations. However, the reduction in cost was  not symmetric since
  applicants, not programs, previously bore most of the costs of in-person
  interviews.    We show that  if doctors can  accept more interviews,
  but the hospitals do not increase the number of interviews they
  offer, then no previously matched doctor is better off and many
  are potentially harmed.  This adverse consequence is the  result of what  we
  call \emph{interview hoarding.} We prove this analytically and characterize optimal mitigation strategies for
  special cases. We use simulations to extend these insights to more
  general settings. 

\end{abstract}

\section{Introduction}
\label{sec:introduction}
Perhaps the most well-known application of matching theory is the
entry-level labor market for  physicians.  In 2021,  37,470
positions were matched through the National Resident Matching Program
(NRMP).  The matching process consists of two steps.  First,
each physician interviews with a set of  residency programs.  Second,
programs and physicians  submit  rank-order lists of those they interview to a
centralized clearinghouse.
  This clearinghouse, run by the NRMP,
matches  physicians to residency
programs using a version of \citeauthor{GaleShapley:AMM1962}'s
(\citeyear{GaleShapley:AMM1962})
Deferred Acceptance (DA)  Algorithm \citep{RothPeranson:AER1999}.

Both programs and applicants are constrained in the
number of interviews they can take part in. Prior to the COVID-19
pandemic, interviews were done in person. These interviews
were particularly costly for physicians since they not only had to bear
travel expenses but also had to take days off from clinical rotations. The cost to
programs was mainly in terms of time.  For the 2020-21
matching season, interviews were conducted virtually. While this
dramatically decreased the cost of interviews for physicians, it did
not substantially change the costs  for the programs. We are interested in the
implications of this asymmetric change on the eventual match.

Intuitively, it seems possible that a doctor might receive a better
  match if she accepts more interviews.  We show a
  surprising result.  As long as she would have been matched with a program, increasing her interview capacity does not help.  Further,
  we show that even if she would have been
   unmatched, she can only benefit if she also would have been part of a pair (with some program) that blocked the original
  matching.\footnote{While the NRMP match is stable with respect
    to submitted preferences over those one interviews, it may not
    be stable with respect to actual preferences over all possible
    matching partners.}   The match rate for US medical school graduates
  is typically around 94\%.\footnote{Specifically, the match rates in 2017
    through 2020 were $94.3\%$, $94.3\%$, $93.9\%$, and $93.7\%$.
    These figures are from the NRMP's \href{https://mk0nrmp3oyqui6wqfm.kinstacdn.com/wp-content/uploads/2021/05/MRM-Results_and-Data_2021.pdf}{``Results and
    Data:  2021 Main Residency Match''}.
  }  Our result says that at least this proportion of the doctors does
  not benefit, but could be harmed, from increasing the number of interviews they participate
  in.

Increasing the number of interviews a doctor accepts has  a negative
  externality: these interviews can no longer be allocated to other
  doctors.  To illustrate, consider a highly sought after physician:
  one who is offered interviews at the leading 
programs and ends up matched with her favorite program.
When  interviews become cheaper, she accepts more
interviews.\footnote{While we assume that agents have complete
  information, if there were an arbitrarily small amount of
  uncertainty about others' preferences, she would accept additional 
  costless interviews.}
However, as she would  already have matched with her favorite program, the
interviews she accepts are from inferior programs.  These interviews
do not help her: she  ultimately matches with the same program as before.
Her additional interviews are, in effect, wasted.  We refer to this as \emph{interview
hoarding}.  Interview hoarding has a cascading affect.  The physicians who
otherwise would have filled these wasted interview slots now interview
with programs they consider inferior.  Such a  physician may have more
interviews, but she does not have better interviews in a precise sense:
she rates every new interview as worse than the program she
matched with before.

This implies a striking result. If
there is an increase to doctors' interview capacities, but
programs do not react, this increase causes the  ultimate match to be
(Pareto) worse from the matched
physicians' perspective.\footnote{The example in
  \cref{sec: motivating example} demonstrates that  there need not
  exist a Pareto ranking from the   programs' perspective.}   These
doctors fall into  three categories:  those 
who hoard interviews worse than their eventual match; those
who receive more but worse interviews; and those who receive
fewer and worse interviews.  The first category is indifferent between
the new outcome and the old.  The latter two categories are harmed.
Even among unmatched physicians, only  those who fall through the cracks---meaning that they are unmatched but would be welcomed by some
programs---could potentially benefit from an increase.  At a match rate of 94\%, this means fewer than 6\%
of  doctors could possibly gain while the overwhelming majority could
be harmed. 

We are not suggesting that there is no benefit to increasing
  interviews.  We are making two unrealistic assumptions.  First, that
  doctors and programs know their preferences perfectly. Second,
  that agents, faced with a constraint on the number of interviews
  they accept, are not strategic.\footnote{We follow the
    approach of \citet{EcheniqueGonzalezWilsonYariv:2020} in assuming
    complete information about preferences and non-strategic offers
    and acceptance of interviews.}  In reality, of course, neither
  doctors nor programs perfectly know their preferences.  The point of
  interviews is to learn more about a candidate or program.
  Similarly,  both doctors and programs will be strategic in
  the interviews they accept.  We expect a doctor to accept an
  interview with a ``safety'' program, one she is sure to match with,
  while lower ranked programs likely do not invite the very best
  candidates for interviews so that they do not ``waste'' their slots.
  Rather, our results show that these are the only two channels
  through which doctors may gain from increasing the number of
  interviews they accept.  By doing so, a doctor  learns about more
  programs and has more flexibility to strategize.\footnote{If a doctor was previously unmatched but
    formed a blocking pair with a hospital, then she must have
    declined that hospital's interview.  Ex-post, this was a strategic
    mistake.  We interpret the doctor's benefit from increasing her
    interviews and being matched as a strategic benefit. }   The
  advantage of our modeling choices is that they allow us to identify a 
subtle bottleneck created by interviews that would likely be lost in the
analysis of a more complex model.

Having shown that increases to doctors' interview capacities has
adverse welfare consequences, we turn to mitigation policies. We consider
policies that limit the numbers of interviews that programs can offer
and candidates can accept. Though there are essentially no such
policies that \emph{always} (for every preference profile) yield a
stable final \assignment (\cref{prop: necessary for global
  sufficient}), we characterize such policies for 
``common preferences'' (\cref{prop: sufficient equal l k}). These are
salient preference 
profiles where every doctor ranks the programs the same way and every
program ranks the doctors the same way. The policies we characterize
are such that there is a common cap on the number of interviews any
program can offer or any candidate can accept. We also show that if the
programs' interview capacities are fixed, say at $l$, then the number
of blocking pairs increases and the match rate decreases as the
doctors' interview cap moves further away from $l$ in either direction (\cref{prop: target
  interview cap}).

Our analytical results can advise policies for more general
settings where preferences are not quite common, but have a common
component. We use simulations
to show that the lessons from our analytical results hold up under
weaker assumptions. The optimal policy is for doctors and programs to
have the same interview capacity.\footnote{This is true whether we  define the optimality of a policy
  as  maximizing the expected proportion of positions that are filled or
  minimizing the expected number of blocking pairs. In fact, these objectives are equivalent. }

\subsection{Related Literature} While there is a large literature on
the post-interview NRMP match,\footnote{See the multitude of papers following \citet{RothPeranson:AER1999}.} there are relatively few papers that
incorporate the pre-match interview process. One of the first to
explicitly model interviews in the classic one-to-one matching model
is \citet{LeeSchwarz:RAND2017}.  In their  model, before participating
in a centralized, two-sided match, firms learn their preferences over
workers by  engaging in costly interviews.  They show that even if
firms and workers interview with exactly the same numbers of agents,
the extent of unemployment in the final match depends critically on
the overlap between the sets of workers that firms interview. Three other
recent papers that incorporate pre-match interviews include 
\citet{Kadam:2015}, \citet{Beyhaghi:2019}, and \citet{EcheniqueGonzalezWilsonYariv:2020}.

Like our paper, \citet{Kadam:2015} considers the implications of loosened
 interview constraints for doctors. However, the focus is on the strategic
allocation of scarce interview slots. For the sake  of tractability, the
analysis is for a stylized model of large markets. Under the
assumption of common 
preferences over programs, they show that
increasing doctors' capacities may increase total surplus, but not in a
Pareto-improving way. Moreover, the match rate decreases. He
also highlights that when preferences are not necessarily common, the
effect is ambiguous, since increased interview capacities dilute doctors'
signaling ability.

\citet{Beyhaghi:2019} also performs a strategic analysis of a
stylized large market model. However, they consider a slightly
different set up with \emph{application} caps for doctors and
interview caps for programs. While similar, application caps are not
 the same as interview caps: they constrain the number of programs a
doctor can express interest in at the outset of the interview matching
phase, but not the number of interviews she can accept at the end. In their model, inequity in the application caps decreases the expected total
surplus. Moreover, when interview capacity is low, low application
caps are socially desirable.

In our model, agents do not choose interviews strategically.
Determining the optimal set of interviews is closely related to the
portfolio choice problems of
\citet{ChadeSmith:Econometrica2006} and \citet{AliShorrer:2021}.
Both of these  solve for the optimal
portfolio when an agent chooses a portfolio of costly, stochastic
options, but  only consumes one of the realizations. In order to
apply the optimal solution  to 
the interview scheduling problem, one would have to  pin down precisely
the  probability of any given pair matching. This is what makes
strategic analysis of the problem intractable without severe
simplifying assumptions (such as those in the papers we have mentioned
above).

As in
\citet{EcheniqueGonzalezWilsonYariv:2020}, which is
methodologically closest to ours, we sidestep this issue. They
explain a puzzling empirical pattern resulting from the NRMP match:
46.3\% of the physicians were matched to their top ranked residency
programs and 71.1\% where matched to a program they ranked in their
top three.  These statistics seem to contradict surveys indicating that
many doctors have similar preferences over residency programs.  
They provide an explanation for this phenomenon by pointing out the
importance of the interviewing process that precedes the match.
Roughly speaking, the pre-match interviewing process restricts the
preferences that the physicians actually submit to the NRMP.
Therefore, a proper interpretation is not that the physicians matched
with their most preferred programs but rather that they matched with
their most preferred programs \emph{among those they interviewed
with}.

Our work is  complementary with these papers in the sense
that they highlight the importance of understanding the prematch
interviews for properly evaluating the NRMP match itself.

\subsection{Motivating Example}\label{sec: motivating example}
We present the intuition behind  the welfare loss from increased
interview capacity for doctors with a simple example. 
Consider a
market with five doctors $\left\{ d_1,\ldots,d_5 \right\}$ and four
hospitals $\left\{ h_1,\ldots,h_4 \right\}$.  Their preferences
are as follows: 
\begin{center}
\begin{tabular}{CCCCC|CCCC}
d_1&d_2&d_3&d_4&d_5& h_1 & h_2 &h_3& h_4\\
\hline
h_1&h_1&h_1&h_3  &h_4& d_1&d_1&d_1 & d_1\\
h_4&h_2&h_2&h_4&h_3  &d_2&d_2&d_4& d_5\\
h_3&h_4&h_3&h_2 &h_1 &d_3&d_4&d_5&d_3\\
h_2&h_3&h_4&h_1& h_2 &d_5&d_3&d_2& d_2\\
&&&&  &d_4&d_5&d_3& d_4
\end{tabular}
\end{center}
Suppose that the interview capacities of the doctors and hospitals are:
\begin{center}
\begin{tabular}{CCCCC|CCC}
d_1&d_2&d_3&d_4&d_5 &h_1 & h_2 &h_3\\
\hline
2&1&1&1&1&2&1&1
\end{tabular}
\end{center}
Interviews are initially offered by hospitals: $h_1$ invites $d_1$ and
$d_2$, and $h_2, h_3,$ and $h_4$ all  invite $d_1$.
As $d_1$ can accept only two invitations, she turns 
$h_2$ and $h_3$ down.  Hospital $h_2$ then offers an
interviews to
 $d_2 $ and $h_3$ invites $d_4$. Since $d_2$ can only accept one
 interview, she declines $h_2$'s invitation. Then, $h_2$ invites and
 is turned down by $d_4$. Finally, $h_2$ invites
 $d_3$ who accepts the invitation. The final
interviews are:
\begin{center}
\begin{tabular}{CCCCC}
d_1&d_2&d_3&d_4& d_5\\
\hline
\{h_1, h_4\}&\{h_1\}&\{h_2\}&\{h_3\}&\varnothing
\end{tabular}
\end{center}

The final matching is computed by applying the doctor-proposing
Deferred Acceptance algorithm 
 to the agent preferences (restricted to agents they interview
with). The outcome is therefore:
\begin{center}
\begin{tabular}{CCCCC}
d_1&d_2&d_3&d_4&d_5\\
\hline
h_1&&h_2&h_3
\end{tabular}
\end{center}

Now suppose each doctor can accept one more invitation but the
hospitals' interview capacities remain the same. In this case $d_1$
does not reject $h_3$'s invitation in the first round of interview
invitations. Similarly,  $d_2$ does not reject $h_2$'s invitation in the
second round.
The interview
schedule is: 
\begin{center}
\begin{tabular}{CCCCC}
d_1&d_2&d_3&d_4\\
\hline
\left\{h_1,h_3,h_4 \right\}&\{h_1,h_2\}&\varnothing&\varnothing& \varnothing
\end{tabular}
\end{center}
This leads to the  final \assignment:
\begin{center}
\begin{tabular}{CCCC}
d_1&d_2&d_3&d_4\\
\hline
h_1&h_2&&
\end{tabular}
\end{center}
Doctors $d_3$ and $d_4$ are worse off despite being able to accept more
interviews. The only doctor who gains is $d_2$.\footnote{The programs, however, are not unanimously better or worse
off: $h_2$ is better off while $h_3$ is worse off.}
We make a few observations about the winners and losers from the
increased interview capacities:
\begin{enumerate}
\item Doctor $d_2$, the only doctor who gained, was among the doctors
  who were originally unmatched.
\item The original matching was not stable. Each of $d_2$ and $d_5$ 
  was part of a blocking pair ($d_2$ blocked with $h_2$ and $d_5$
  blocked with $h_4$). Specifically, the only doctor to gain was among
  doctors who blocked.
\item Despite being part of a  pair that blocked the original
  matching, $d_5$ did not regret turning down any interview
  invitations, while $d_2$ did.\footnote{By ``regret'' we mean that
    she turned down an interview invitation from a hospital that she
    preferred to her final matching.}  In particular, the
  only doctor who gained \emph{did} previously reject an invitation that she
  ended up regretting.
\end{enumerate}
These observations are not specific to this example. We
show that   (\cref{thm: more not
  better}), in general, a doctor can only
gain from an increase to interview capacities if she was originally
unmatched,  she was part of a blocking pair,
\emph{and} she regretted rejecting an interview invitation.

\section{The Model}\label{sec: market}
A \textbf{market}
consists of a triple {\boldmath $(D,H,P)$}, where $D$ is a finite set of
\textbf{doctors},  $H$ is a finite set of \textbf{hospitals}, and $P$
is a profile of strict \textbf{preferences} for the doctors and
hospitals.  For each $h\in H$, $\mathcal P_h$ is the
set of strict preferences over $D\cup\{h\}$, and for each $d\in D$,
$\mathcal P_d$ is the set of strict preferences 
over $H\cup\{d\}$.
The set of preference profiles is
$\mathcal P \equiv \times_{i\in H\cup D} \mathcal P_i$.

There are two phases to the matching process. The first is a
decentralized interview phase and the second is the centralized  matching phase. The former involves  many-to-many
matching  while the latter is a standard one-to-one matching problem \citep{RothSotomayor:1990}.

A \textsl{many-to-many} matching is a
function 
 \mbox{$\nu: H \cup D \rightarrow 2^{H\cup D}$} such that, for each
 $d\in D$ and $h\in H$, $\nu(d) \subseteq H$, $\nu(h) \subseteq D$, and $h \in \nu(d)$ if
and only if $d\in \nu(h)$.

For each $h\in H$, let ${ \i_h}\in \mathbb N$ be
\textbf{\boldmath $h$'s interview   capacity}. Similarly, for each
$d\in D$, let $\k_d\in \mathbb N$ be \textbf{\boldmath 
  $d$'s interview
  capacity}. We call the profile $(\i,\k) = ((\i_h)_{h\in H}(\k_d)_{d\in D})$ the
  \textbf{interview \arrng}. An \textbf{interview \assignment} is a many-to-many
  matching $\nu$ such that for every doctor $d$, $|\nu(d)|\leq
  \k_d$ and for every hospital $h$, $|\nu(h)|\leq \i_h$.
  
 An interview \assignment $\nu$ is 
  \textbf{pairwise stable} if there is no
  doctor-hospital pair $(d,h)$ such that $h\not\in \nu(d)$ but:  
  \begin{itemize}
	\item either $|\nu(h)|<\i_h$ and $d\mathrel P_h h$ or there exists a $d'\in \nu(h)$ such that $d \mathrel{P_h}  d'$, and
	\item either $|\nu(d)|<\k_d$ and $h\mathrel P_d d$ or there exists a $h'\in \nu(d)$ such that $h \mathrel{P_d} h'$.
\end{itemize}
  
A matching is a function $\mu:H \cup D \rightarrow H \cup D$ such that
 $\mu(h)\in D\cup\{h\} $, $\mu(d)\in H \cup\{d\}$, and $\mu(d)=h$ if
 and only if $\mu(h)=d$. We say that the matching $\mu$ is
 \textbf{individually rational} if for each $i\in D\cup H, \mu(i)
 \mathrel R_i i$.
The pair $(d,h)$ 
 \textbf{blocks}  the matching $\mu$ if $h\mathrel P_d \mu(d)$
and $d\mathrel P_h \mu(h)$.  A matching is \textbf{stable} if it is
individually rational and is not blocked by any pair.

To describe how the market works, we follow the approach of
\citet{EcheniqueGonzalezWilsonYariv:2020} by assuming complete
information and non-strategic behavior.\footnote{In
  \cref{sec:role-interviews-play}, we discuss a 
  version of our model and some of our results when preferences are
  formed during the interviews.} This means that hospitals na\"ively make
offers to their most preferred doctors and these offers, if rejected,
trickle down to less preferred doctors. Thus, given $(\i,\k)$ and
$P\in \mathcal P$, the  final
\assignment, which we call the
\mbox{\textbf{\boldmath$(\i,\k)$-\assignment}}, is the outcome of the
following two phase
process:\footnote{We only differ from
  \citet{EcheniqueGonzalezWilsonYariv:2020} in that we set the
  interview matching to be the \emph{hospital-optimal} many-to-many stable
  matching, while they set it to be the \emph{doctor-optimal} one. Their
   choice is appropriate for the question they ask. However, we have chosen
  to approximate the decentralized interview phase through the hospital-proposing
  DA. This difference does not drive  our results, as we explain in
  \cref{sec:doct-optim-interv}.} 
\begin{enumerate}
\item [Phase 1:] The interview \assignment $\nu$ is the
  hospital-optimal pairwise stable
  many-to-many matching where the capacities of the hospitals
  and doctors are given by $\i$ and $\k$, respectively. This can be
  computed by applying  the hospital-proposing
  deferred acceptance (DA) algorithm: each $h\in H$ is matched with up
  to $\i_h$ doctors 
  and each $d\in D$ is matched with up to $\k_d$ hospitals. Since we
  ignore the informational aspect of the 
  problem, the input to DA is a choice function for each agent that is
  responsive to her preference relation and constrained by her
  interview capacity.\footnote{For the sake of completeness, we define in
    \cref{sec:choice-funct-interv} the acceptant
    and responsive choice functions that we appeal to while running DA
    to compute the interview \assignment.
    According to these choice functions, hospitals na\"ively select  their
    most preferred doctors from each set. In \cref{sec:2021-match}, we consider alternative choice
    functions that reflect common heuristics. Simulations involving
    such a heuristic based choice  are closer to the results of
    the 2021 match, but the qualitative effects are similar. The
    simpler assumption of na\"ive behavior renders the results more
    transparent.}
  The hospital-proposing DA algorithm is an approximation of the 
  decentralized process by which hospitals invite doctors,
  extending invitations to further doctors when invitations are
  declined.
\item [Phase 2:] The $(\i,\k)$-\assignment  is
  computed by applying the    doctor-proposing DA algorithm. The input
  to DA is the true 
  preference profile restricted to the 
  interview match, $(P_i|_{\nu(i)})_{i\in D\cup H}$.
\end{enumerate}

The deferred acceptance algorithm is used twice.  To avoid confusion, we refer to running the $\ida$ and the $\mda$ algorithms.

\section{Welfare Impact of Increased Interviews}\label{sec: main}
Our aim is to study how a change in the interview costs impacts a
market.  It is intuitive that when doctors interview with more
hospitals that the interviewing market becomes more competitive.
However, it is not clear what the impact on a doctor's
ultimate assignment is.  We expect a doctor to benefit from her
increasing her interviews but to be harmed when other doctors also
increase the interviews; therefore, the ultimate impact is
ambiguous. 

We show that in fact only certain doctors who were
previously unmatched can possibly benefit from additional
interviews.\footnote{As a reminder, we are assuming that doctors have
  perfect information and are nonstrategic.  While doctors can benefit
  from the increased information and easing of strategic constraints,
  our result say that these are the only ways a doctor can 
  benefit.}  Specifically, to gain from 
additional interviews, a doctor must have been previously unmatched,
been part
of a blocking pair for the original matching, \emph{and} regretted
turning down an interview invitation.\footnote{By
  ``regretted'' we mean that this unmatched doctor rejected an
  interview proposal in the first phase from a hospital that she would
  have matched with if she had not rejected it.}

 In 2020, 93.7 of
doctors graduating from a US medical schools where matched to a
program by the NRMP.\footnote{This number is from the NRMP's \href{https://mk0nrmp3oyqui6wqfm.kinstacdn.com/wp-content/uploads/2021/05/MRM-Results_and-Data_2021.pdf}{``Results and
    Data:  2021 Main Residency Match''}. The
  2020 match rate was in fact slightly lower than in previous years.
  In 2016 through 2020, the match rates were $93.8\%$, $94.3\%$,
  $94.3\%$, $93.9\%$, and $93.7\%$, respectively.}  Therefore, few doctors could
benefit from increasing interviews while potentially many could be
harmed.

\begin{theorem}\label{thm: more not better} Suppose that for each
  $d\in D,\; \k_d'\geq \k_d$, and let $\mu$ and $\mu'$ denote the
  $(\i,\k)$-matching and $(\i,\k')$-matching, respectively. If
  $\mu'(d)\mathrel P_d \mu(d)$ then $\mu(d) = d$, $(d, \mu'(d))$ blocks
  $\mu$, and for each $h\in \nu(d), h\mathrel P_d \mu(d)$.

  That is, a doctor
  can only gain from increased interview capacities prefer if she 
  \begin{inparaenum} [(i)]
\item  was unmatched,
\item was part of a blocking pair, and 
\item regretted turning down an interview from a hospital she blocked with. 
\end{inparaenum}

\end{theorem}

Stability of the original matching is a natural notion
of equilibrium in  a well functioning market. An implication of
\cref{thm: more not better} is that if such an equilibrium were
shocked with increased interview capacities, the consequence would be
a Pareto-worsening of the match from the doctors' perspective.

\begin{corollary}\label{cor: more not better} Suppose that for each
  $d\in D,\; \k_d'\geq \k_d$, and let $\mu$ and $\mu'$ denote the
  $(\i,\k)$-matching and $(\i,\k')$-matching, respectively.  If $\mu$
  is  stable, then for every $d\in D$, $\mu_d \mathrel{R_d}
  \mu'_d$.\end{corollary}

We show a series of lemmas to prove \cref{thm: more not better}. In what follows, let $\nu$ and $\mu$ be the interview and final matchings, respectively, under $(\i,\k)$. 
Similarly, let $\nu'$ and $\mu'$ be the interview and final matchings
under $(\i,\k')$. We frame the temporal language below in reference to
a hypothetical change  in doctors' 
interview capacities from $\k$ (``before'') to $\k'$ (``after'').  As
a reminder, we use $\ida$ and $\mda$ to refer to the algorithms for
computing the interviews and the final matching respectively. 

We first establish a number of properties of the interview
matchings. The intuition for these results comes from one of the
classical results in two-sided matching theory:  When 
the set of men increases, no man benefits from this increased competition while no woman is harmed.\footnote{See Theorem 2.25 of \citet{RothSotomayor:1990}.}  In our setting, an
increase in the number of interviews a doctor accepts plays the role of additional men participating in the market. This means that the hospitals are able to interview
better doctors. However,  there is a tension between
interviewing better doctors and interviewing the ``right''
doctors. Thus improving the set of candidates a hospital interviews 
does not necessarily translate to an improvement in its eventual match. 

We do not use our first lemma  directly in the proof of \cref{thm:
  more not better}, but it is the key to understanding the effect of additional interview capacities on the interviewing outcomes.

\begin{lemma}Suppose that for each $d\in D,\; \k_d'\geq \k_d$.  If
  $h\in \nu(d)$ (hospital $h$ interviews   $d$ under capacities
  $\k$), then $d$ does not reject $h$ when  $\ida$ is run with capacities
  $\k'$. That is, no doctor rejects a hospital that previously
  interviewed her. \label{norejects} 
\end{lemma}
\begin{proof} Suppose not.  When $\ida$ is run (with capacities $k'$),
  let $d$ be the first doctor to reject a hospital $h$ that
  interviewed her under capacities $\k$.  As $\k'_d \geq \k_d$, $d$ must have received a new
  interview proposal from some hospital $h'$.  As $h'$ did not propose to $d$
  when capacities where $\k$, it must have been rejected by some doctor $d'\in
  \nu(h)$, a doctor it previously interviewed.  But
  this contradicts $d$ being the first doctor to reject a hospital she
  previously interviewed with. 
\end{proof}

We cannot say whether a doctor ``prefers'' her interviews under $\k$
versus $\k'$ as we only have a doctor's preferences over individual
hospitals and not sets of hospitals.  However, we  show---in a
specific sense---that while a doctor may get new interviews, she does
not get better interviews. 

\begin{lemma}\label{nobetter} Suppose that for each $d\in D,\; \k_d'\geq \k_d$.  For every $d\in D$, if $h'\in
  \nu'(d)\setminus \nu(d)$, then $\forall h\in \nu(d),
    h\mathrel P_d h'$.  That is, any new interview a doctor receives is worse than all of her prior interviews.
  \end{lemma}	
\begin{proof}
  Suppose $d$ receives an interview proposal from some $h\notin
  \nu(d)$.  If $h$ did not propose an interview to $d$ under $\k$, then $h$ must
  have been rejected by a doctor that it previously
  interviewed. However, this contradicts \cref{norejects}.  If $h$ did
  propose an interview to $d$ under $\k$, then since $h\notin\nu(d)$,
  $d$  rejected $h$'s interview proposal. By revealed preference, for each $h'\in \nu(d)$, $h'\mathrel P_d h$. 
\end{proof}

In particular,  if a doctor was previously matched to a hospital, then every new interview she receives is worse than her previous assignment.  In the classical result, no man benefits from the increased
competition due to additional men and also no woman is harmed.  An
analogous result holds in our framework.  A hospital either has the
same set of interviews, additional interviews, or it interviews new
doctors it prefers to its previous interviews.  In each of these scenarios,
the hospital's set of interviews  (weakly) improves.  The next lemma
shows that whenever a program interviews a new doctor, the program
``keeps'' all the interviews with doctors it prefers.

\begin{lemma}\label{oldisnew}
Suppose that for each $d\in D,\; \k_d'\geq \k_d$.  For every $h \in H$, if $d'\in \nu'(h)$,
$d \in \nu(h)$, and  $d \mathrel P_h d'$, then $d \in \nu'(h)$. That is, if a hospital interviews a doctor $d$, it interviews every
doctor it used to interview  among those that it prefers to $d$.
\end{lemma}

\begin{proof}
  
Since $d \mathrel P_h d'$, $h$ proposes an interview to $d$ before it
proposes an interview  to $d'$ under $\k'$.  By
\cref{norejects}, $h$ is not rejected by any doctor it previously
interviewed.  As $h$ proposes to $d'$ under $\kappa'$, it must have already
proposed to but  not have been rejected by $d$.  Therefore, $h$ continues to
interview $d$.  
\end{proof} 

Having established the above properties of the interview matching, we
can now consider the \assignment phase. We start by showing that
rejections in $\mda$ are monotonic with
regards to doctors' interview capacities.

\begin{lemma}
  Suppose  doctor $d$ is rejected by
  hospital $h$ in when $\mda$ is run with interviews $\nu$. If $h\in
  \nu'(d)$, then $h$ rejects $d$ when $\mda$ is run with interviews
  $\nu'$.  That is, a hospital continues to reject any doctor that it previously rejected.
  \label{rejection monotonicity}
\end{lemma}

\begin{proof}
We  prove a stronger statement.  We  prove that if $d$ is
rejected by $h$ in round $n$ when $\mda$ is run under $\nu$,
and if $h\in \nu'(d)$, then $h$ rejects $d$ in round $n$ or earlier
when $\mda$ is run under $\nu'$.  We proceed by induction on $n$.
For the base step, consider a doctor $d$ who was rejected by hospital
$h$ in the first 
round of $\mda$ under $\nu$.  Let $d'$ be the doctor that
$h$ tentatively 
accepted when she rejected $d$.  By assumption, \mbox{$d\in \nu'(h)$.}
By \cref{oldisnew}, 
since $h$ prefers $d'$ to $d$ and $h$ interviews $d$ under $\k'$, $h$
also interviews $d'$ under $\k'$ ($d'\in \nu'(h)$).  Moreover, by
\cref{nobetter}, any new interview $d'$ receives is worse for her
than $h$.  Therefore, $d'$ still  proposes to $h$ in the first
round under the new 
capacities and $h$ still rejects $d$ in favor of a doctor it finds at
least a good as $d'$.

Our inductive hypothesis is that for any $d\in \nu'(h)$, any $n\geq2$,
and any $m<n$, if $h$  rejected $d$ in round $m$ of $\mda$ under
$\nu$, then $h$ rejects $d$ in round $m$  of  $\mda$ under $\nu'$ or earlier.

First we show that for any doctor $d$ and any hospital $h$ that
interviews $d$ under both $\k$ and $\k'$ ($h\in \nu(d) \cap
\nu'(d)$), if $d$ proposes to hospital $h$ in round $n$ or earlier when
$\mda$ is run under $\nu$, then $d$ proposes to $h$ in round $n$  of
 $\mda$ or earlier under $\nu'$.  Consider any $h' \in \nu'(d)$ such that $h'
\mathrel{P_d} h$.   By \cref{nobetter}, $h'$ cannot be a new interview
($h'\in \nu(d)$).  Since when $\mda$ is run under $\nu$, $d$ proposes to
$h$ in round $n$ or earlier, $h'\in \nu(d)$, and $h' \mathrel{P_d} h$,
$d$  proposes to and is rejected by $h'$ prior to round $n$.  By the
inductive hypothesis, when $\mda$ is run under $\nu'$, $h'$
rejects $d'$ prior to round $n$.  Therefore, $d'$ proposes to $h$ by
round $n$. 

To complete the induction step,  suppose that hospital $h$ rejected
doctor $d$ in favor of 
doctor $d'$ in round $n$ under $\nu$.  By assumption, $d \in \nu'(h)$. By
\cref{oldisnew}, since $h$ continues to interview $d$ but prefers
$d'$, $h$ also continues to interview $d'$.  Since $h\in \nu(d) \cap
\nu'(d)$ and $h\in \nu(d') \cap \nu'(d')$ and both $d$ and $d'$
propose to $h$ in round $n$ of $\mda$ or earlier  under $\nu$, we have
shown that both $d$ and $d'$ propose to $h$ by round $n$ in $\mda$
under $\nu'$.  Thus, by round 
$n$, under $\nu'$, $h$ receives a proposal it prefers to
$d$.  Therefore,  $h$  rejects $d$ under
$\nu'$  in round $n$ or  earlier. 
\end{proof}

We are now ready to prove \cref{thm: more not better}.

\begin{proof}[Proof of \cref{thm: more not better}]
Consider any doctor $d\in D$.  First, suppose that $\mu(d)=h \in H$.
Consider any $h'$ such that $h' \mathrel{P_d} h$.  We  show that
$\mu'(d) \neq h'$.  If $h' \not\in \nu'(d)$, then we are done.
If $h' \in \nu'(d)$,  by \cref{nobetter}, $h' \in
\nu(d)$ (all new interviews are worse than $h$).  Since
$h'\mathrel{P_d} h$, $h'\in \nu(d)$,  and $\mu(d)=h$, when $\mda$ is run
under $\nu$, $d$ proposed to $h'$ and was rejected.  By
\cref{rejection monotonicity}, when $\mda$ is run under $\nu'$, $h'$ 
rejects  $d$.  Therefore, $\mu'(d)\neq h'$.  We conclude that no
doctor who was matched by $\mu$ is matched to a better
hospital by $\mu'$. 

Now suppose $\mu(d)=d$ but that $\mu'(d)=h\in H$.  We  show that
$d$ and $h$ block $\mu$.  If $\mu(h)=h$, then since $\mu'$ is
individually rational, $d$ and $h$ block $\mu$.  If
$\mu(h)=d' \in D$, since no doctor who was matched by $\mu$ is matched
to a better hospital by $\mu'$ and  since $\mu'(d') \neq h=\mu(d')$,
we deduce that $h
\mathrel{P_{d'}} \mu'(d')$.  If $d'\mathrel  P_h d$, then by
\cref{oldisnew}, $d'\in \nu'(h)$ ($h$ interviewed $d'$ before, so since
$h$ now interviews doctor $d$ whom it likes less, it continues to
interview $d'$).  But this contradicts the stability of $\mu'$ at the
restricted preferences  since
$d'$ and $h$ interview under $\nu'$ and prefer each other to their
respective assignments.  Therefore, it must be that $d \mathrel P_h d'$, and
indeed, $d$ and $h$ block $\mu$. 
Finally, we show that $d$ regretted rejecting $h$. Since $h$ and $d$
block $\mu$ and $d$ was unmatched, $d\notin \nu(h)$. Moreover,
since  $d\mathrel{P_h} \mu(h)$, $h$ proposed to $d$ when $\ida$ was
run under capacities $\k$ and $d$ rejected $h$: for each $h'\in
\nu(d), h'\mathrel P_d h$.
\end{proof}

\cref{thm: more not better} tells us that doctors increasing
the number of interviews they accept has little scope for improving doctor welfare but great potential for harm.
  The  example in \cref{sec:
  motivating example} 
illustrates  that only certain unmatched doctors can gain from
increased capacities. This example is not pathological.
\cref{norejects} and \cref{nobetter} highlight  the root cause of the 
inferior match, which is interview hoarding.
The set of doctors who escape the
adverse effects of an increase to capacities  are a subset of
unmatched doctors. When the match rate is high, this set is small.

Under our simplifying assumptions that agents are non-strategic and
have complete information, \cref{thm: more not better} implies that
the shift to virtual interviews for the 2020-21 season of the NRMP
ought to have led to an 
inferior matching. In \cref{sec:2021-match}, we contrast simulation results for
this na\"ive behavior with the common heuristic of including a ``safety'' candidate when choosing a set. While the NRMP has touted the high
match rate for 2021, this may be driven by hospitals being matched
to safety candidates (under  heuristic behavior) rather than being
unmatched (under na\"ive behavior). Other than the match rate, heuristic
choice does not qualitatively affect the  results.

\section{\Arrngs That Ensure Stability }
As stated in \cref{cor: more not better}, when the initial \arrng
leads the 
two-phase process to a stable \assignment, no doctor benefits from
increased capacities.
Thus, it is important in making policy choices related to interview
capacities to understand what capacity profiles lead to a stable
matching. Of course, the
answer   depends on specifics of the market, such as
the ratio of doctors to hospitals and how correlated or aligned
preferences are.  However, we are able to provide tight
characterizations for certain ``end-point'' cases that provide
intuition for more general markets.

\subsection{Stability for All Preferences}
\label{sec:glob-suff-regim}

In studying stability of the two-phase process, we first discuss worst case performance:
what \arrngs yield stable matchings \emph{for every} preference profile? It
turns out that only very extreme \arrngs satisfy this property. We
characterize these \arrngs in our next result.

\begin{proposition}
  \label{prop: necessary for global sufficient}
  \Arrng $(\i, \k)$ yields a stable matching for all preferences if and only if  either
  \begin{enumerate}
  \item   every doctor and every hospital has only unit interview
    capacity---that is, for each $d\in D,  \k_d = 1$
    and for each $h\in H, \i_h = 1$---or
  \item every doctor and every hospital has high   interview capacity---that
    is, for each $d\in D,  \k_d \geq \min\{|D|,|H|\}$
    and for each $h\in H, \i_h \geq \min\{|D|,|H|\}
    $.
  \end{enumerate}
\end{proposition}
\begin{proof} This result is trivial if $|D|=1$ or $|H|=1$, so suppose
  that $|D|\geq 2$ and $|H|\geq 2$.
  We first prove necessity. Suppose that $(\i,\k)$ is yields a stable matching for all preferences.

    We start by establishing that if one doctor or hospital has greater than
    unit interview capacity, then every doctor and hospital has
    interview capacity of at least two. Stated differently,
    if any doctor or hospital has unit
    capacity, then all doctors and hospitals have unit capacity. We
    denote by $\nu$ the interview \assignment and by $\mu$ the
    $(\i,\k)$-\assignment. 

    \begin{claim}\label{claim: small capacity}
    \begin{enumerate}      
    \item If there is  $d\in D$ such that     $\k_d > 1$, then
      for each $d'\in D, \k_{d'} \geq 2$ and for each $h\in H, \i_h \geq
      2$, and 
    \item If there is  $h\in H$ such that     $\i_h > 1$, then
      for each $h'\in H, \i_{h'} \geq 2$ and for each $d\in D, \k_d \geq
      2$.
    \end{enumerate}
  \end{claim}
  \begin{proof}
    We prove only the first statement as the proof of the second
    statement is analogous---it requires only a reversal of the roles
    of doctors and hospitals.

    Suppose, for the sake of contradiction,  that $(\i,
    \k)$ always yields a stable matching and there are $d_1\in
    D$ such that $\k_{d_1} >1$ and  $h_2\in H$ such that $\i_{h_2} =
    1$. Let $h_1 \in H\setminus\{h_2\}$ and $d_2 \in D\setminus
    \{d_1\}$. Consider $P\in \mathcal P$ where each doctor ranks $h_1$
    first and $h_2$ second, and each hospital ranks $d_1$ first and
    $d_2$ second.  All hospitals offer  an interview to $d_1$ and as
    $\k_{d_1}>1$, $d_1$ accepts interviews from at least $h_1$ and
    $h_2$.  Since $\i_{h_2}=1$, $h_2$ only interviews  $d_1$.  Let
    $\mu$ be the $(\i,\k)$-\assignment.  Since $(\i,\k)$ always yields
    a stable matching, $\mu$ is stable, so
    $\mu(d_1)=h_1$, as $h_1$ and $d_1$ are mutual favorites.  Therefore,
    $\mu(h_2)=h_2$ as $h_2$ only interviews  $d_1$.  Note that
    $(d_2,h_2)$ forms a blocking pair of $\mu$ as $h_2
    \mathrel{P_{d_2}} \mu(d_2)$, since $\mu(d_2) \not \in \left\{ h_1,h_2
    \right\}$, and $d_2 \mathrel{P_{h_2}} h_2$.  This
    contradicts the stability of $\mu$ and thus the assumption that
    $(\i,\k)$ always yields a stable matching. We have therefore established
    that if there is $d\in D$ such that $\k_d > 1$, then for each
    $h\in H$, $\i_h \geq 2$.

    We now prove that if there is a $d_1\in D$ such that $\k_{d_1} > 1$,
    then for each $d\in D$, $\k_d\geq 2$. Suppose for the sake of
    contradiction that there is 
    $d_2\in D$ such that $\k_{d_2} = 1$. Let $h_1,h_2\in H$.    
    Consider $P\in \mathcal P$ such that each doctor ranks $h_1$
    first and $h_2$ second, and each hospital ranks $d_1$ first and
    $d_2$ second.  As we have shown above, $\i_{h_1}, \i_{h_2} \geq
    2$, so both $h_1$ and $h_2$ offer interviews to both
    $d_1$ and $d_2$. Since $h_1$ is her favorite hospital, 
    $d_2$ accepts  its offer. Thus, $\nu(d_2) = \{h_1\}$.
    However,  $\mu(d_1) = h_1$ since $d_1$ and $h_1$ are mutual
    favorites, so $\mu(d_2)=d_2$.  This means that $(d_2,h_2)$
    form a blocking pair of $\mu$ as the only hospital $d_2$ prefers
    to $h_2$ is $h_1$.  This contradicts the stability of $\mu$ and
    thus the assumption that $(\i,\k)$ always yields a stable matching. 
      \end{proof}

We complete the proof of necessity by showing that neither a doctor nor a
hospital can have an intermediate capacity. 
\begin{claim}\label{claim no intermediate caps}
There is no $d\in D$ such that $1<\k_d <\min\{|D|,|H|\}$, and there is
no hospital $h$ such that $1<\i_{h}<\min\{|D|,|H|\}$.   
\end{claim}

\begin{proof}
We prove this statement for the case where  $|D| \leq |H|$. The proof
when $|H| < |D|$ is symmetric.

Suppose for the sake of  contradiction that $d_1\in D$ is  such that $\k_{d_1} = k$
where $1 < k < |D|$. Let $P\in\mathcal P$ be such that  for $i$ from 1 through $k+1$:
\begin{align*}
P_{d_1} &: h_2,h_3,\ldots ,h_{k+1},h_1,\ldots \\
P_{h_i} &: h_i,h_1,\ldots ,h_{i-1},h_{i+1},\ldots \\
P_{h_1} &: d_1,d_2,\ldots \\
P_{h_i} &: d_i,d_1,\ldots ,d_{i-1},d_{i+1},\ldots
\end{align*}
We have
constructed the 
preference profile $P$   such that:
\begin{itemize}
	\item For each $i$ from 1 through $k+1$, $d_i$ and $h_i$ are matched in every stable \assignment.
	\item Each  of the $k+1$ hospitals  $h_1,\ldots, h_{k+1}$  offers $d_1$ an interview.
	\item Doctor $d_1$ accepts  interview offers from hospitals
          $h_2,\dots, h_{k+1}$, but not from $h_1$. 
\end{itemize}
The first and third points are immediate consequences of the
preferences.  The second is a consequence of the
first part of 
\cref{claim: small capacity}: since $\k_{d_1}>1$, every hospital has
an interview capacity of at least two and ranks $d_1$ in its top two.  However, this contradicts the
definition of $\mu$ as the $(\i,\k)$-\assignment, since $h_1\notin
\nu(d_1)$ yet by stability, $h_1 = \mu(d_1)$.
  
A similar construction shows that there is no $h\in H$ such
that 
$1<\i_{h}<|D|$.  Suppose for the sake of  contradiction that $h_1\in H$ is such
that $\i_{h_1} = l$ where $1 < l < |D|$.
Let  $P\in\mathcal P$ be such that for $i$ from 1 through $l+1$:
\begin{align*}
P_{d_1} &: h_1,h_2,\ldots \\
P_{d_i} &: h_i,h_1,\ldots ,h_{i-1},h_{i+1},\ldots \\
P_{h_1} &: d_2,d_3,\ldots ,d_{l+1},d_1\\
P_{h_i} &: d_i,d_1,\ldots ,d_{i-1},d_{i+1},\ldots
\end{align*}
By the second part of \cref{claim: small capacity}, since
$\i_{h_1}>1$, every doctor has a capacity of at least two.  Therefore:
\begin{itemize}
	\item For each $i$ from 1 through $l+1$, $d_i$ and $h_i$ are matched in every stable \assignment.
	\item Each  of the $l$ doctors  $d_2,\ldots, d_{k+1}$
          accepts an interview from  $h_1$.
	\item Hospital $h_1$ does not offer  $d_1$ an   interview.
\end{itemize}
Thus, $h_1\notin \nu(d_1)$, so $h_1\neq \mu(d_1)$. This contradicts
the stability of $\mu$, the $(\i,\k)$-\assignment, and in turn the
assumption that $(\i,\k)$ always yields a stable matching.
\end{proof}

We now turn to sufficiency. If every agent has an interview capacity
of one, then the interview \assignment is actually a
matching. Moreover, it is a stable matching. So, suppose that each
agent has an interview capacity of at least $\min\{|D|, |H|\}$. If
$|D| = |H|$, then the interview \assignment involves an interview
between every mutually acceptable doctor-hospital pair. This means
that the $(\i,\k)$-\assignment is the doctor optimal stable matching
under unrestricted preferences, which is stable. We now show, that
even if $|D| < 
|H|$ or $|D| > |H|$, the $(\i,\k)$-\assignment, $\mu$, is stable. 
Suppose the doctor-hospital pair $(d,h)$ blocks
$\mu$. By defintion of $\mu$ as the $(\i,\k)$-\assignment, if
$h\mathrel P_d \mu(d)$ and $d \mathrel P_h \mu(h)$, then $h\notin
\nu(d)$.

Suppose $|D|<|H|$. Since
$\i_h \geq |D|$, $h$ would have offered an interview to $d$ and would have
been rejected when $\ida$ is run, so $\nu(d)$ contains $\k_d$ hospitals that $d$ prefers to
$h$. Since $h \mathrel P_d \mu(d)$, and $\mu(d)\in \nu(d) \cup\{d\}$,
this means $\mu(d) = d$. Then, $d$ is rejected by every hospital in
$\nu(d)$ when $\mda$ is run. However,
$|\nu(d)| = \k_d \geq  |D|$ and since $d$ is acceptable to every
hospital in $\nu(d)$, she is only rejected when another doctor
applies. However, this implies that when DA terminates in the matching
phase, every hospital 
in $\nu(d)$ has tentatively accepted some doctor other than $d$, which
is a contradiction---there are not enough such doctors.

Suppose $|H| < |D|$. Since $\k_d \geq |H|$, $d$ does not reject any
interviews she is offered. Since $h\notin \nu(d)$,  $h$ offers
interviews to and has them accepted by $\i_h \geq |H|$ doctors whom it
prefers to $d$. Since $d\mathrel P_h \mu(h)$, $h$ does not receive a
proposal from any $d'\in \nu(h)$ when $\mda$ is run since it finds all
such $d'$  better than $d$. This 
implies that  each $d'\in \nu(h)$ is tentatively accepted by some
hospital other than $h$ when DA terminates, which is a
contradiction---there are not enough such hospitals.
\end{proof}

\cref{prop:  necessary for global sufficient} highlights a previously
overlooked role that the interview phase plays in determining whether
or not the ultimate NRMP match is stable.  While interviews are
necessary for agents to gain information, we learn from \cref{prop:
  necessary for global sufficient} that interviews can also act as a
bottleneck.  Even with complete information, once any agent is capable
of participating in more than one interview, all agents must interview with essentially the entire market to be certain that the ultimate match is stable.

\subsection{Homogeneous \Arrngs}
\label{sec:homogeneous}
One potential intervention that has been suggested to deal with
interview hoarding is a cap on the number of interviews each doctor can
accept.\footnote{For instance, see
\cite{MorganWinkelEtAl:JSE2020}.}  Here
we consider  \textbf{homogeneous \arrngs}:  all 
doctors face the same cap and all hospitals face the same cap. Thus,
the intervention 
would be described by two numbers: an interview capacity $l\in \mathbb
N$ for hospitals and an interview capacity $k\in \mathbb N$ for doctors.
The pair $(l,k)$ corresponds to the \arrng $(\i,\k)$ where for each
$h\in H, \i_h = l$ and for each $d\in D, \k_d = k$. 

By \cref{prop:  necessary for global sufficient} a 
homogenous \arrng $(l,k)$  always yields a stable matching only if $l = k =
1$ or $l,k\geq \min\{|D|,|H|\}$. Nonetheless, $(l,k)$ \emph{may} yield
stable matchings for  specific profiles of preferences. One might ask
  whether, starting at a profile $P\in \mathcal P$ and  
  \arrng $(l,k)$ that  yields a stable matching at $P$, the comparative statics
  with respect to $l$ and $k$ are consistent. The following examples
  demonstrate that this is not so.
It may be that, depending on $P$, increasing $k$ 
renders a  previously stable matching unstable, or the
opposite. In other words, the effect of an increase 
to $k$ is specific to $P$ and $l$. 

\begin{example}\label{example:comparative statics}
Either incrementing or
 decrementing  $l$ or $k$ can  create or eliminate instability.
\end{example}

\noindent Suppose  $|D| = |H| = 3$ and consider  $P\in \mathcal
P$ such that for each $i=1,2, 3$,\footnote{This can be
  embedded into a larger problem instance.} 
\[
  \begin{array}{ccc}
    \begin{array}{c}
      P_{h_i}   \\
      \hline
      d_1\\
      d_2\\
      d_3\\
      h_i
    \end{array}&&
        \begin{array}{c}
      P_{d_i}   \\
      \hline
          h_1\\
          h_2\\
          h_3\\
          d_i
    \end{array}
  \end{array}
\]For $P$,  $(2,2)$ yields a stable matching: the interview \assignment is $\nu$
such that  $\nu(h_1) =
\nu(h_2) = \{d_1,d_2\}$ and $\nu(h_3) = \{d_3\}$. So, the
$(l,k)$-\assignment is $\mu$ such that for
each $i=1,2,3$, $\mu(h_i)= d_i$, which is the unique stable matching.

We now observe that  if we increment or decrement
either $l$ or $k$  by one, the matching is no longer stable. In
other words, none of $(1,2)$, $(3,2)$, $(2,1),$ or $(2,3)$ yields a
stable matching. We summarize the interview \assignment and the $(l,k)$-\assignment for each of these
below.
\[
\hspace{-0.25in}  \begin{array}{|c|c|c|}
    \hline (l,k)&\text{ interview \assignment }& \text{ $(l,k)$-\assignment}\\
    \hline (1,2) & \nu(h_1) = \nu(h_2) = \{d_1\}, \nu(h_3) = \{d_2\} &
                                                                       \mu(h_1)
                                                                       =
    d_1, \mu(h_3) = d_2, \mu(h_2) = h_2, \mu(d_3) = d_3\\
    \hline (3,2) & \nu(h_1) = \nu(h_2) = D, \nu(h_3) = \{\} &\mu(h_1)
                                                              = d_!,
                                                              \mu(h_2)=d_2,
    \mu(h_3) = h_3, \mu(d_3) = d_3\\
    \hline (2,1) & \nu(h_1) = \{d_1,d_2\}, \nu(h_2) = \{d_3\},
                   \nu(h_3) = \{\} & \mu(h_1) = d_1, \mu(h_2) = d_3,
                                     \mu(h_3) = h_3, \mu(d_2) = d_2\\
    \hline (2,3) & \nu(h_1) = \nu(h_2) = \nu(h_3) = \{d_1,d_2\} &
                                                                  \mu(h_1)
                                                                  =
                                                                  d_1,
    \mu(h_2)= d_2, \mu(h_3) = h_3, \mu(d_3) = d_3\\
    \hline
  \end{array}
\]All four of the $(l,k)$-matchings are unstable. \hfill $\circ$

\medskip
The mechanics of \cref{example:comparative statics} are robust
and it is not by accident that $(2,2)$ yields a stable outcome to
start with. The preferences in the example have a particularly salient configuration,
which we focus on here.
A profile
$P\in \mathcal P$ has \textbf{common 
  preferences} if all doctors rank the hospitals
in the same way, and all hospitals rank the doctors in the same
way. To further restrict the definition, we also require
   that each doctor finds each hospital acceptable and each
hospital  finds each doctor acceptable. That is, for each pair $d,
d'\in D$ and each pair  $h, h'\in H$,  $P_d|_H = P_{d'}|_H$, $P_h|_D =
P_{h'}|_D$,  $d \mathrel P_h h$, and 
$h\mathrel P_d d$.\footnote{Under common
preferences there is a unique stable matching.}

As we see from \cref{example:comparative statics}, a result like 
\cref{prop: necessary for global sufficient} does not hold if we
restrict ourselves to common preferences. Our next result is a
characterization of homogeneous \arrngs that yield stable matchings for common
preferences.\footnote{The characterization of \cref{prop: sufficient
    equal l k} does not hold for \arrngs that are not homogeneous. For
a counterexample, suppose  $|D| = 4$, $|H| = 3$,  there is
$d\in D$ such that $\k_d = 3$, for each $d'\in D\setminus \{d\},
\k_{d'} = 2$, there is $h\in H$ such that $\i_h = 4$,
and for each $h'\in H\setminus \{h\}, \i_{h'} = 2$. For any common
preferences, $(\i,\k)$ yields a stable matching.}

\begin{proposition}\label{prop: sufficient equal l k}
  Let $P\in \mathcal P$ be such that there are common
  preferences. A homogeneous  \arrng $(l,k)$ yields a stable matching at $P$ if and only
  if $l = k$ or $l,k\geq \min\{|D|,|H|\}$.
\end{proposition}
\begin{proof}
Let $\{d_t\}_{t=1}^{|D|}$ and
  $\{h_t\}_{t=1}^{|H|}$ be enumerations of $D$ and $H$, respectively,
  such that every hospital prefers $d_t$ to $d_{t+1}$ and every doctor 
  prefers $h_t$ to $h_{t+1}$. Let $m = \min\{|D|, |H|\}$.
  There is a unique stable matching $\mu^*$, such that for each
  $t=1, \dots, m$, $\mu^*(h_t) = d_t$.

  Let $\nu$ be the interview \assignment under $(l,k)$ and $\mu$ be the $(l,k)$-\assignment.

First, we show that $(l,k)$ yields a stable matching at $P$ only if $l =
k$ or  $l,k\geq \min\{|D|,|H|\}$. Suppose $l\neq k$. If  $l < k$ and
$l < \min\{|D|,|H|\}$, then for each $t = 1,\dots, k$, $\nu(h_t) =
\{d_1, \dots, d_l\}$. In particular, $d_k \notin \nu(h_k)$ so
$\mu(h_k) \neq d_k$. On the other hand, if 
$l > k$ and
$k < \min\{|D|,|H|\}$, then for each $t = 1,\dots, l$, $\nu(d_t) =
\{h_1, \dots, h_k\}$. In particular, $h_l \notin \nu(d_l)$ so
$\mu(d_l) \neq h_l$. In either case, the $(l,k)$-\assignment is
not stable.

Now, we show that if $l = k \leq m$, then $(l,k)$  yields a
stable matching.
For each $t = 1, \dots, m$, let $\underline t = \lfloor
\frac{t-1}{l}\rfloor$. Then, for each $t = 1,\dots,
i$, $\nu(h_t) = \{d_{\underline t +1},\dots, d_{\underline t
  +l}\}$ and $\nu(d_t) =  \{d_{\underline t +1},\dots, d_{\underline t
  +l}\}$. Thus, for each $t=1, \dots, m$, $\mu(h_t) = d_t$ and so
$\mu$ is stable.

Finally, if $l,k \geq m$, then for each $t = 1,\dots,m$, $\nu(d_t)
\supseteq\{h_1,\dots, h_m\}$. Since \mbox{$h_t\in \nu(d_t)$,} $h_t =
\mu(d_t)$ and so $\mu$ is stable.
\end{proof}

If the hospitals' interview capacity is fixed at some specific $l$,
where to set the doctors' interview
cap, $k$, is an important policy decision. \cref{prop: sufficient equal l k} says that the optimal
value for $k$ is exactly at $l$ whether  the objective is to minimize the
number of blocking pairs or to maximize the match rate (the proportion
of positions that are filled). Our next result sheds light on this
objective.

\begin{proposition}\label{prop: target interview cap}
Fix the hospitals' interview capacity at $l$ and consider $k$ and $k'$
such that 
either $k' < k \leq l$ or  $l \leq k < k'$. Suppose $P\in
\mathcal P$ has common preferences. The $(l,k')$-\assignment has more
blocking pairs and a weakly lower match rate than the $(l,k)$-\assignment.
\end{proposition}
\begin{proof}
  Let $P\in \mathcal P$ be such that there are common
  preferences. Let $\{d_t\}_{t=1}^{|D|}$ and
  $\{h_t\}_{t=1}^{|H|}$ be enumerations of $D$ and $H$, respectively,
  such that every hospital prefers $d_t$ to $d_{t+1}$ and every doctor 
  prefers $h_t$ to $h_{t+1}$.

  Let $m = \min\left\{\left\lfloor \frac{|H|}{k}\right\rfloor,
    \left\lfloor \frac{|D|}{l}\right\rfloor\right\}$. The interview
  matching is such that for each $d_t$, if  $t \leq ml$,
  \[
    \nu(d_t) = \{h_{(n-1)k+1},\dots, h_{nk}\}\text{ where $n$ is such
      that }(1-n)l < t \leq nl
  \]
  if  $ml < t \leq (m+1)l$,
  \[
    \nu(d_t) =\left\{
      \begin{array}{ll}
        \{h_{mk+1},\dots, h_{n}\}&\text{if }|H| \geq mk+1\\
        \emptyset & \text{otherwise}
      \end{array}\right.
      \text{ where } n=\min\{|H|, (m+1)k\}
  \]
  and if $(m+1)n < t$, $\nu(d_t) = \emptyset$.

  We first consider the case where  when $k> l$ and show that the
  number of matched hospitals is decreasing in $k$ and that the
  number of blocking pairs is increasing in $k$.
  
  Given $P$ and its restriction to $\nu$,  the
  $(l,k)$-\assignment, $\mu$, at $P$ is such that for each $d_t$, if  $t \leq ml$,
  \[
    \mu(d_t) = h_{(n-1)k+(t\bmod l)},  \text{ where $n$ is such
      that }(n-1)l < t \leq nl,
  \]
  if  $ml < t \leq (m+1)l$,
  \[
    \mu(d_t) =\left\{
      \begin{array}{ll}
        h_{mk+(t\bmod l)}& \text{if }|H| \geq mk+(t\bmod l)\\
        d_t & \text{otherwise},
      \end{array}\right.
  \]
  and if $(m+1)l < t$, $\mu(d_t) = d_t$.

  Let $n = \min\{|H| - mk, |D|-ml\}$. Given the $(l,k)$-\assignment
  above, the set of matched hospitals is 
  \[
    \{h_{ik+s}: i=0,\dots,m-1, s=1\dots, l\}\cup\{h_t: t = mk+1,\dots,mk+n\}.
  \]
  Therefore, the number of matched hospitals is $ml
  +n$. Holding $l$ fixed, this is decreasing in $k$.

  The $(l,k)$-\assignment is blocked by all pairs  consisting
  of an unmatched hospital and any doctor with a higher index.
  That is, $(h_t, d_{t'})$ such that
  $t \leq mk$,  $t-1\bmod k \geq l$  and $t' > t$. These are
  the only pairs that block it. 
  Thus, the number of blocking pairs is
  \[
    \sum_{n=0}^{m-1}\sum_{i=l+1}^k |D| - (nk+i).
  \]Holding $l$ fixed, this  is increasing in $k$.

  Now, we consider the case where $k < l$ and show that the number of
  matched hospitals is increasing in $k$ and the number of blocking
  pairs is decreasing in $k$.

  Given $P$ and its restriction to $\nu$,  the
  $(l,k)$-\assignment at $P$ is such that for each $h_t$, if  $t \leq mk$,
  \[
    \mu(h_t) = d_{(n-1)l+(t\bmod k)},  \text{ where $n$ is such
      that }(n-1)l < t \leq nl,
  \]
  if  $mk < t \leq (m+1)k$,
  \[
    \mu(h_t) =\left\{
      \begin{array}{ll}
        h_{ml+(t\bmod k)}& \text{if }|D| \geq ml+(t\bmod k)\\
        h_t & \text{otherwise}
      \end{array}\right.
  \]
  and if $(m+1)k < t$, $\mu(h_t) = h_t$.

    Let $n = \min\{|H| - mk, |D|-ml\}$.  Given $(l,k)$-\assignment above, the set of
    matched hospitals is $\{h_t:t\leq mk+n$. Therefore, the number of
    matched hospitals is $mk+n$. Since $k < l$, this is weakly
    increasing in $k$.
    
  The $(l,k)$-\assignment is blocked by all pairs  consisting
  of an unmatched doctor and any hospital with a higher index.
  That is, $(d_t, h_{t'})$ such that
  $t \leq ml$,  $t-1\bmod l \geq k$  and $t' > t$. These are
  the only pairs the block it. 
  Thus, the number of blocking pairs is
  \[
    \sum_{n=0}^{m-1}\sum_{i=k+1}^l |D| - (nl+i).
  \]Holding $l$ fixed, this  is decreasing in $k$.
  \end{proof}

\section{Simulations}
\label{sec:simulations}
Our analytical results are of two sorts. On one hand, \cref{thm: more
  not better} applies without restrictions on 
preferences. However, it only helps identify the problem caused by
increases to doctors' interview capacities without suggesting a remedy. On the other hand, when we focus on common
preferences, \cref{prop: sufficient equal l k,prop: target interview cap} deliver a clearcut
policy prescription. In this section, we use simulations to bridge the
gap. This allows us to consider how changes in the
doctors' interview capacities affect hospitals' welfare, match rates,
 stability, and so on, in a more general setting.

While there is evidence that
preferences do indeed have a common component
\citep{Agarwal:AER2015,Rees-Jones:GEB2018}, agents care about ``fit''
as well. Moreover, an idiosyncratic component is to be expected. 
We adopt the random utility model of
\citet{AshlagiKanoriaLeshno:JPE2017}.\footnote{This, in turn, 
  is adapted from \citet{HitschHortacsuAriely:AER2010}.}
Each hospital $h\in H$  has a common component to its quality, $x_h^C$, and a
``fit'' component, $x^F_h$. Similarly, each doctor
$d\in D$ has a common component to her quality, $x_d^C$ and a fit
component, $x^F_d$. The utilities that $h$ and $d$ enjoy from being
matched to one another are
\[
  \begin{array}{c}
  u_h(d) = \beta x_d^C - \gamma\left(x_h^F -x_d^F\right)^2 +
    \varepsilon_{hd}\\
\text{ and }\\
  u_d(h) = \beta x_h^C - \gamma\left(x_h^F -x_d^F\right)^2 +
    \varepsilon_{dh},
  \end{array}
  \]
  respectively, where $\varepsilon_{hd}$ and $\varepsilon_{dh}$  are
  drawn independently from the standard  logistic distribution.  Each
  $x_h^C, x_h^F, x_d^C,$ and $x_d^F$ is drawn independently from the
  uniform distribution over $[0,1]$.
  The coefficients $\beta$ and $\gamma$ weight the common and fit
  components, respectively. When $\beta$ and $\gamma$ are both zero,
  preferences are drawn uniformly at random. As $\beta\to \infty$,
  these approach  common preferences. As $\gamma$
  increases, preferences become more  ``aligned'': the fit, which is
  orthogonal to the common component, becomes more important.

  
  \newcommand{\simhospcap}{25}
  \newcommand{\simbeta}{40}
  \newcommand{\simequivalpha}{0.864}
  \newcommand{\simgamma}{20}
  \newcommand{\simtrials}{100}
  \newcommand{\simnumdocs}{470}
  \newcommand{\simnumhosps}{400}
  \newcommand{\simlowk}{2}
  \newcommand{\simhik}{100}

  \newcommand{\toptier}{200}
  \newcommand{\opttierk}{13} 
  
  \newcommand{\simkeqlblockingpairs}{136.08}
  \newcommand{\simkeqlmatchrate}{99.9525}
  

  Our simulated market has $\simnumhosps$
  hospitals.\footnote{The NRMP match is broken down into smaller
    matches by specialty. In 2020, among 50 specialties for PGY-1
    programs, the largest had 8,697 positions, the 10th largest had 849
    positions, the 25th largest had 38 positions, the 49th largest had
    one position, and the smallest had
    no positions. This data is available from the
    \href{https://mk0nrmp3oyqui6wqfm.kinstacdn.com/wp-content/uploads/2020/06/MM_Results_and-Data_2020-1.pdf}{NRMP}. Our
    chosen number of hospitals is comparable to the 70th percentile
    among specialties (that is, 70\% of specialities are smaller than this). 
    }
We have set the number of doctors at $\simnumdocs$.\footnote{There  
 were, on average, 0.85 PGY-1 positions per applicant in the 2020. Our 
 chosen number of doctors reflects this ratio. }
The parameters for the random utility model are  $\beta = \simbeta$ and $\gamma = 
  \simgamma$. 
Since our interest is in the effects of
changes to doctors' interview capacities, we fix  hospital interview
capacities at $l = \simhospcap$.  We discuss the robustness of our
findings with regards to our choices of model and
parameter values in \cref{apx:sim}.

Our first simulation results involve varying $k$ from $\simlowk$ to
$\simhik$.\footnote{We have chosen this upper bound to be high enough
  that further increases have little effect. Thus, we interpret this as
  doctors being essentially unconstrained in how many interviews they
  can accept.}   \cref{fig: k v match rate} shows that the match rate
increases and then 
decreases. On the other hand, \cref{fig: k v blocking pairs} shows
that the number of blocking pairs decreases and then increases.
These results are consistent with \cref{prop: target interview cap}.
Despite preferences not being common, both the match rate and
stability (measured by the number of blocking pairs) are optimized at $k=l$.

\begin{figure}
  \centering
  \begin{subfigure}[b]{0.49\textwidth}
    \centering
    \ifimages
    \includegraphics[scale=0.45]{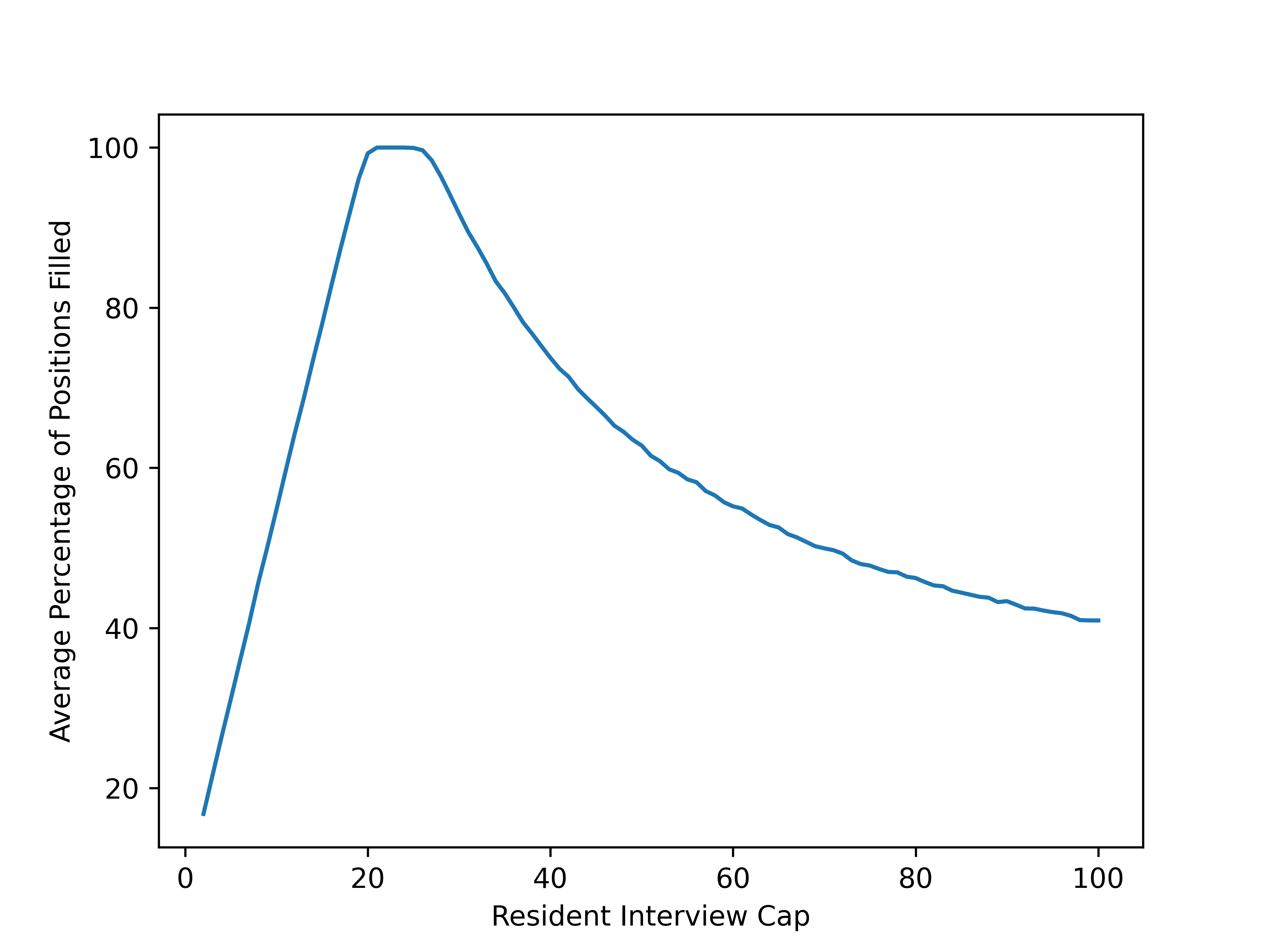}
    \fi 
    \caption{The average
      match rate is    $\simkeqlmatchrate$ when
      \mbox{$k=l$ }
.}
    \label{fig: k v match rate}
  \end{subfigure}
  \hfill
  \begin{subfigure}[b]{0.49\textwidth}
    \centering
    \ifimages
    \includegraphics[scale=0.45]{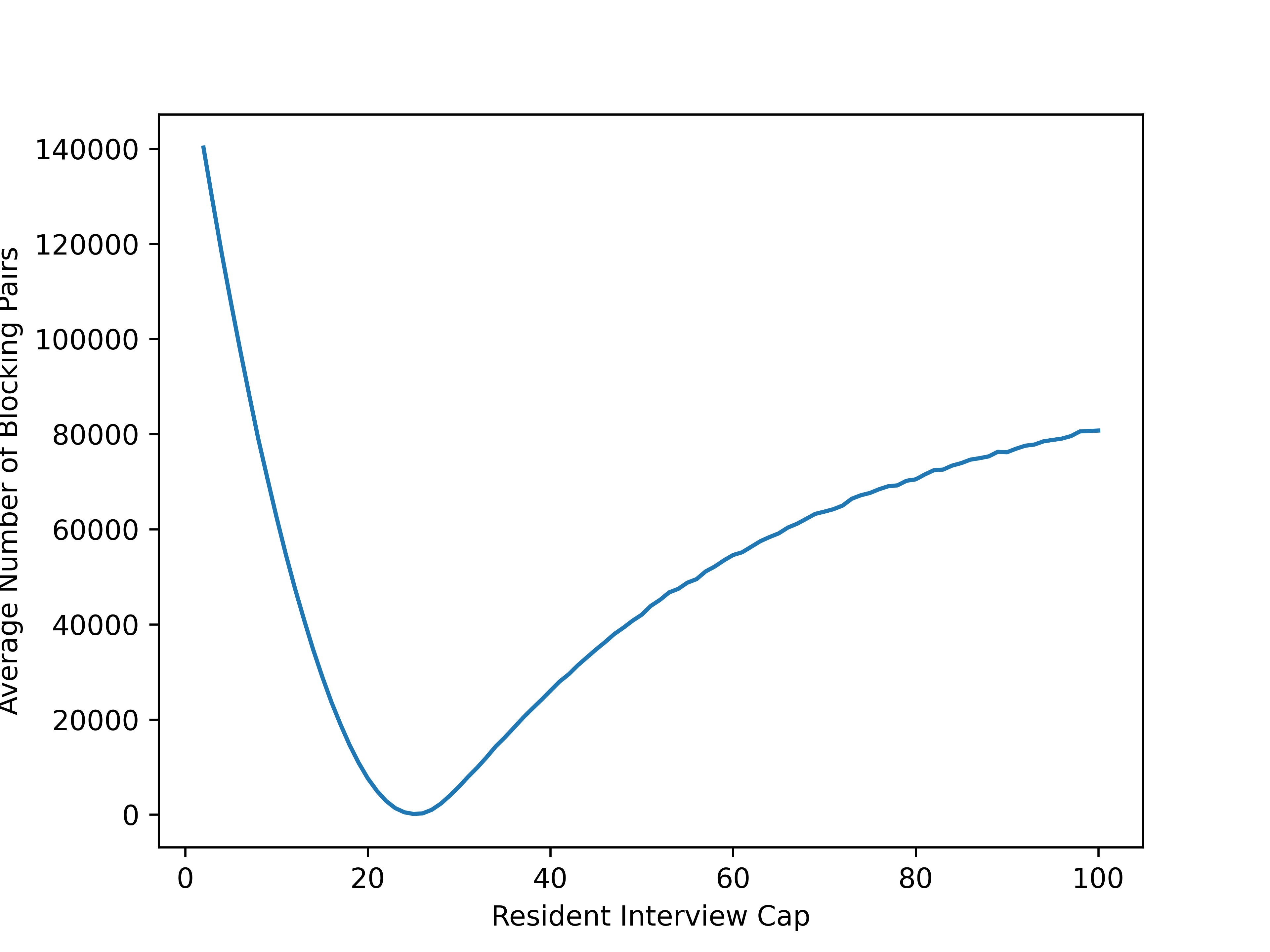}
    \fi
    \caption{The average number of
      blocking pairs is $\simkeqlblockingpairs$ when $k=l$.}\label{fig: k v blocking pairs}
  \end{subfigure}
  \caption{We vary $k$ from $\simlowk$ to $\simhik$ with $l$ fixed at $\simhospcap$.}
\end{figure}

Our next set of results evaluate a hypothetical policy of restricting
doctors to a maximum of $k = l(=\simhospcap)$ interviews. We compare this policy with the benchmark of no
intervention where doctors are
completely unconstrained. 

\cref{fig: residents cap preference} shows the distribution of the
number of doctors who prefer their match under the 
    optimal $k$ over the benchmark as well as the distribution of
    those  with the opposite preference. We see that the former
    is considerably higher than the latter.  \cref{thm:
      more not better} allows for certain unmatched doctors to gain
    from relaxing this policy. However,
    \cref{fig: residents cap preference}
    shows that such doctors are rare: on average, only 0.17 doctors gain while 334.45 are harmed.
\cref{fig: hospitals cap preference} shows the same  distributions,
except for hospitals.  Despite the fact that
    \cref{thm: more not better} does not address hospitals' welfare,
    our simulations show that more hospitals prefer the optimal cap of
    $k=l$
    than leaving the 
        doctors unconstrained.
       The policy also has the benefit bringing
stability to the final matching. \cref{fig: blocking
  diff} shows the distribution of the  increase in the number of
blocking pairs 
    when
    we go from the policy of $k=l$ to doctors being unconstrained. Finally, we compare the distribution of
    interviews among the doctors between the two \arrngs in \cref{fig:
    interview dist}. The constraint limiting doctors to  $k = l$
  interviews binds for many doctors.  One implication  is that significantly more
    doctors receive zero interviews when they are unconstrained.
    This is consistent with the intuition that if interviews were
    costless for doctors, then highly sought after doctors would hoard
    interviews and others would be left with nothing.

\begin{figure}
  \centering
  \begin{subfigure}[b]{0.49\textwidth}
    \centering
    \ifimages
    \includegraphics[scale=0.45]{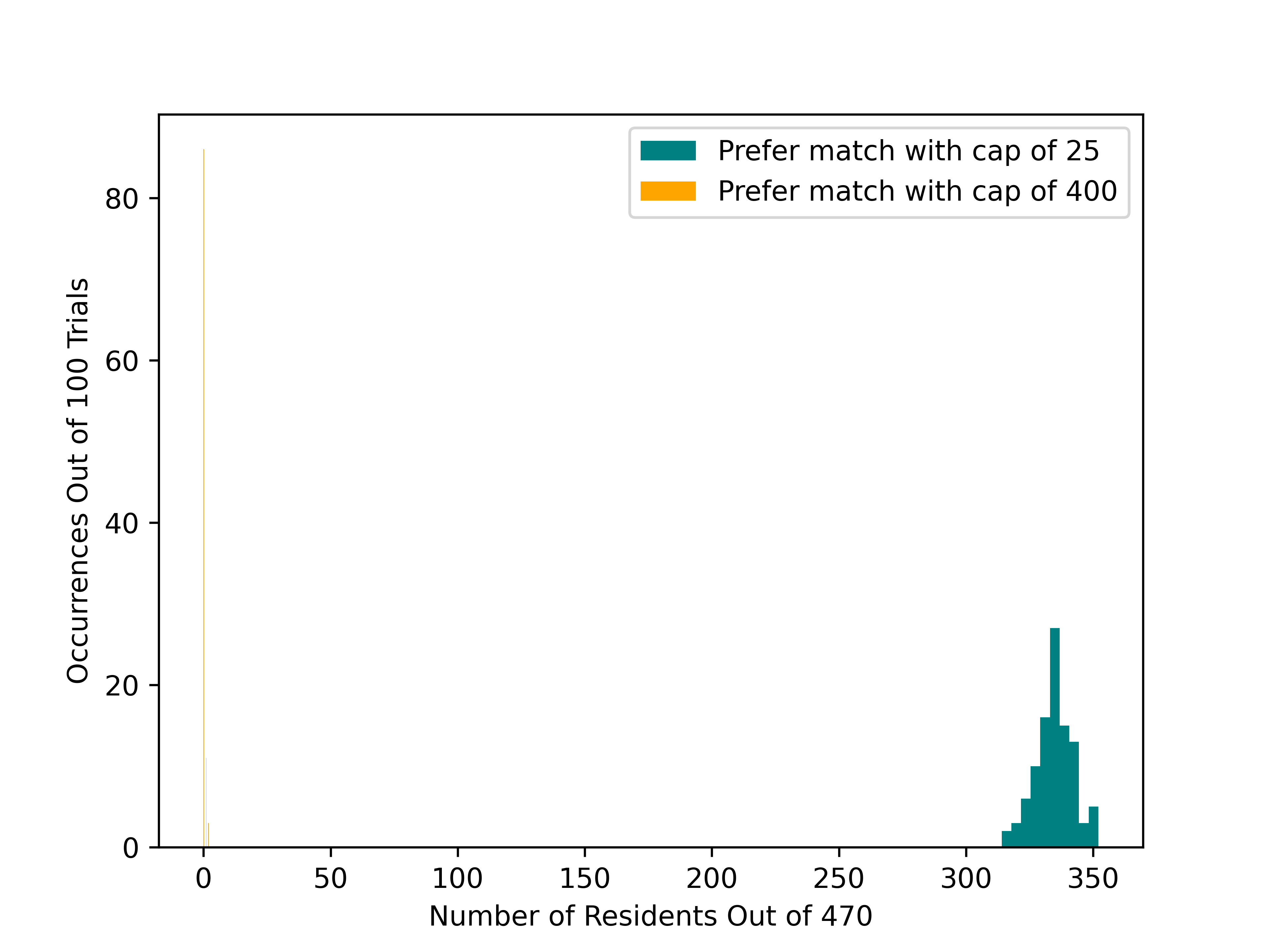}
    \fi
    \caption{Distribution of the number of doctors who prefer their
      match at $k = l $ over being unconstrained and vice versa.}
    \label{fig: residents cap preference}
  \end{subfigure}
  \hfill
  \begin{subfigure}[b]{0.49\textwidth}
    \centering
    \ifimages
    \includegraphics[scale=0.45]{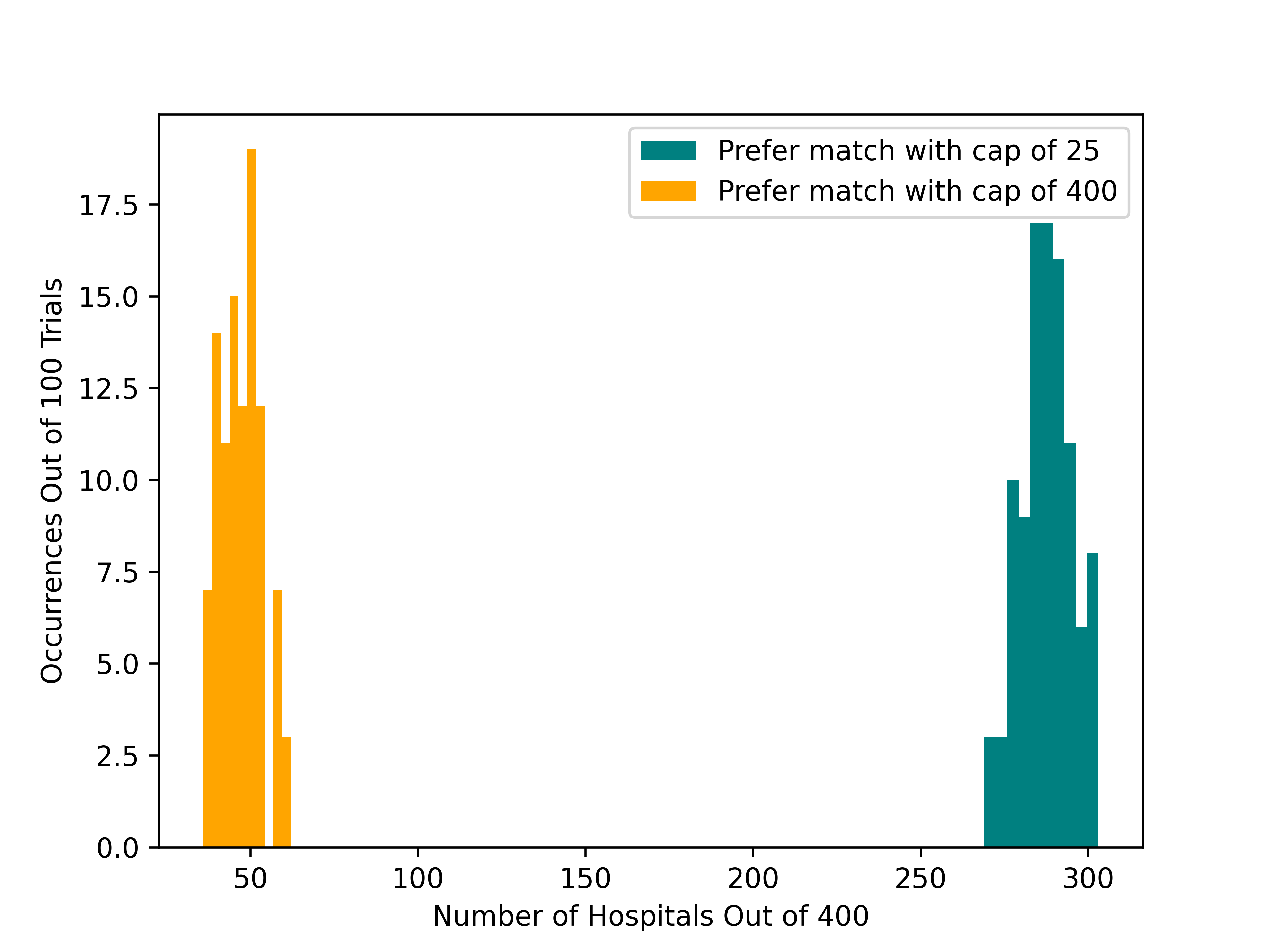}
    \fi
    \caption{Distribution of the number of hospitals that prefer their
      match at $k = l$ over the doctors being unconstrained and vice versa.}
    \label{fig: hospitals cap preference}
  \end{subfigure}
  \hfill
  \begin{subfigure}[b]{0.49\textwidth}
    \centering
    \ifimages
    \includegraphics[scale=0.45]{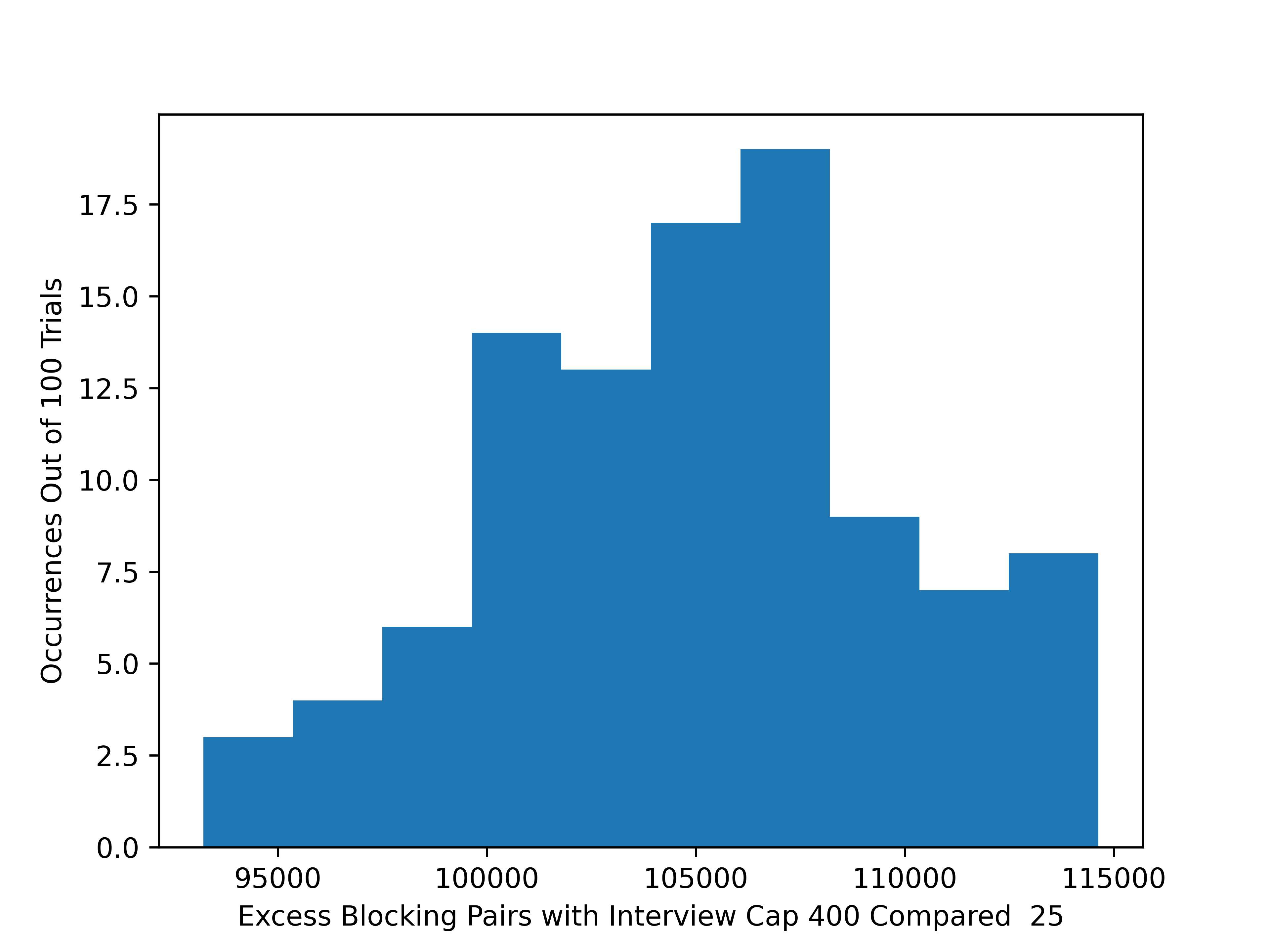}
    \fi
    \caption{Distribution of the number of excess blocking pairs when
      doctors are unconstrained over the number of such pairs at $k=l $.}
    \label{fig: blocking diff}
  \end{subfigure}
  \hfill
  \begin{subfigure}[b]{0.49\textwidth}
    \centering
    \ifimages
    \includegraphics[scale=0.45]{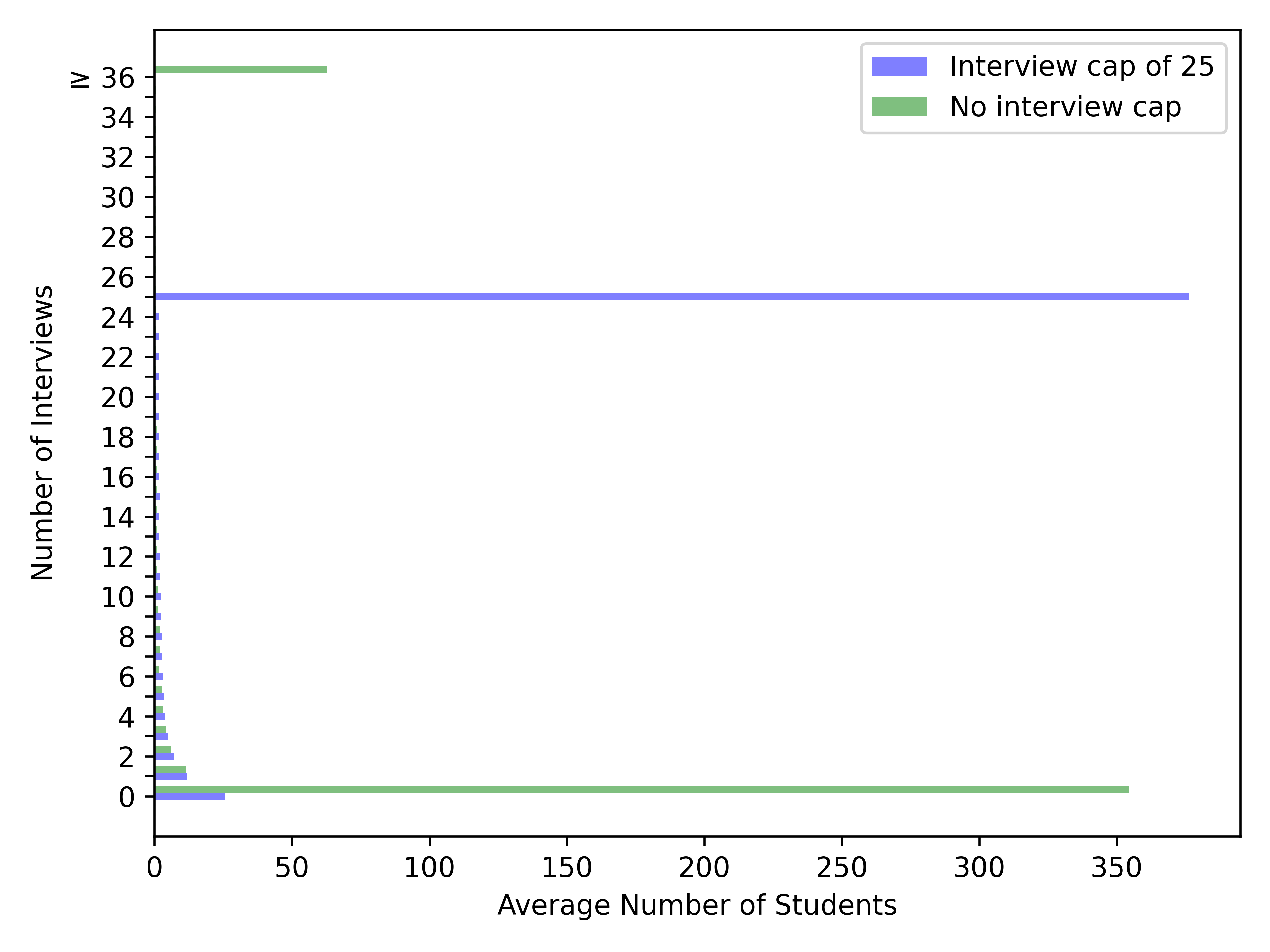}
    \fi
    \caption{Distributions of interviews at $k=l $ and without a
      cap. The uncapped distribution vanishes with the
    number of interviews reaching the hundreds.}
    \label{fig: interview dist}
  \end{subfigure}
  \caption{Comparisons of the intervention of capping doctors'
    interview capacities at $k =l $ to leaving doctors
    unconstrained. }
\end{figure}

In our last simulations , we consider the possibility that the NRMP
could set not only 
a cap on interviews that doctors can accept, but can  also control the number of
interviews that hospitals offer. From  \cref{prop: sufficient equal l
  k}, we know that if preferences are common, the match rate would be
maximized where $k=l$. \cref{fig: match rate l k} shows that, even
when preferences are 
not exactly common, this is still the optimal policy: along the
diagonal, where $k$ and $l$ are equal, the match rate is close to
100\% and  the match is almost stable.

\begin{figure}
  \centering
      \ifimages
      \includegraphics[scale=1]{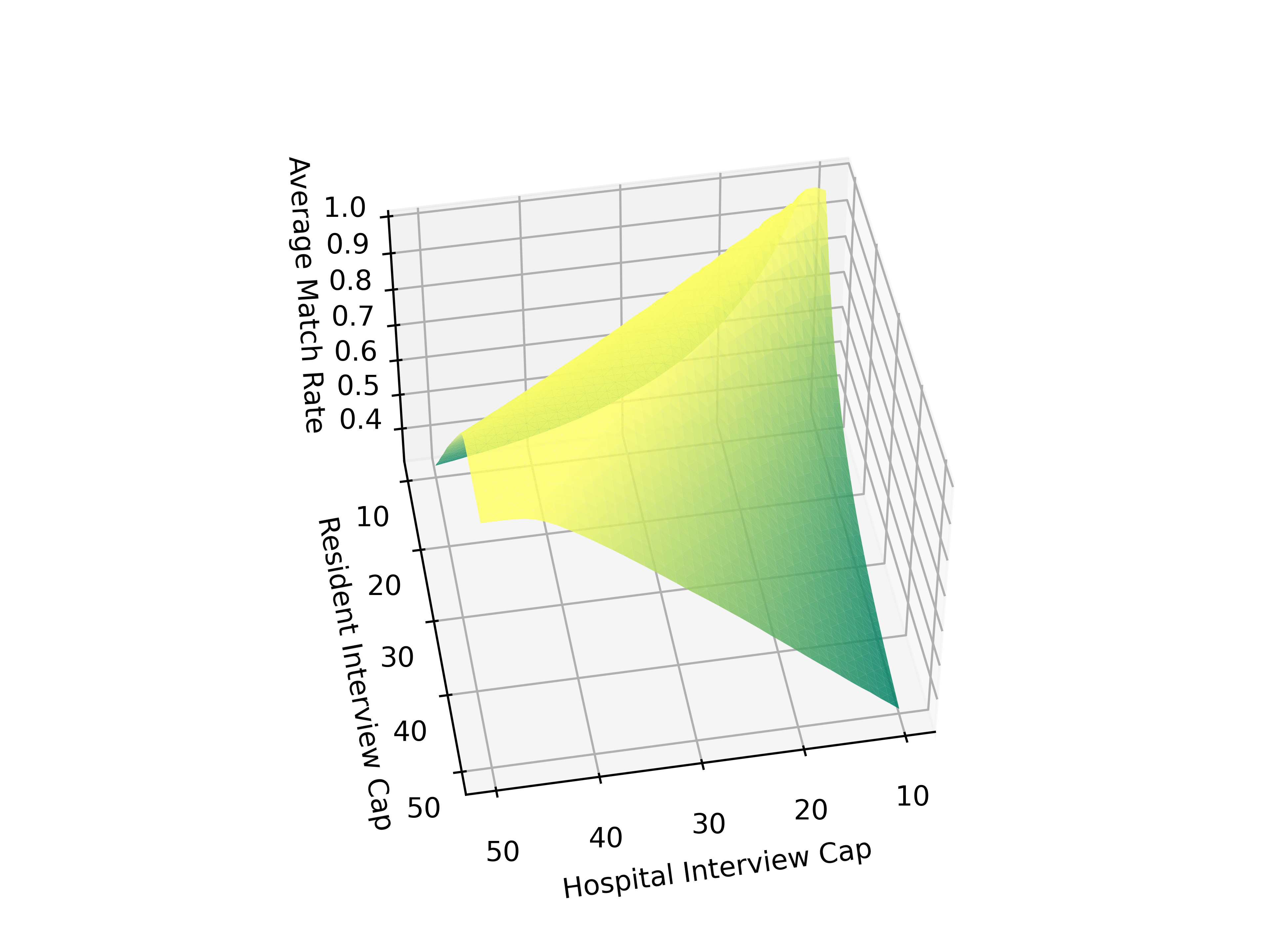}
      \fi
      \caption{Match rate as a function of $l$ and $k$}\label{fig: match
    rate l k}
\end{figure}

In \cref{sec:2021-match}, we modify our model to consist of two tiers of doctors
and contrast simulations under na\"ive choices by hospitals (as above)
to those under a heuristic for hospitals where every choice involves a
``safety'' doctor from the lower tier. The only qualitative difference
is in the match rate: rather than being unmatched, hospitals match
with safety candidates. The welfare based comparisons are very
similar.

\section{Conclusion}

The COVID-19 pandemic has had a significant impact on the way
interviews are conducted.  It has also impacted the
distribution of interviews among doctors.  Through our
theoretical results and simulations, we argue that the 2021 NRMP Match
was likely  inferior to previous years.

In future years, the NRMP should consider policies to mitigate these
effects. Our analysis supports the idea of interview caps
and   our simulations provide evidence that such a policy would
reduce the bottleneck created by the interview phase.  Such caps
can be implemented with very limited centralization, for instance,
using a ticket system.

Even if such interventions are not possible in the very short run,
our policy prescription is that residency programs should be advised
to increase the number of candidates they interview  relative to
previous years.

Design of a fully centralized clearinghouse for interviews is an area that remains
open. As earlier work on the 
interview pre-markets have shown, strategic analysis is only
tractable under very stringent assumptions
\citep{Kadam:2015,LeeSchwarz:RAND2017,Beyhaghi:2019}. Nonetheless,
the current paper adds to the evidence (along with
\citet{EcheniqueGonzalezWilsonYariv:2020}) that a 
more holistic approach that includes the interview stage is critical.

This interview-driven bottleneck is likely a factor in other matching
contexts as well, including fully decentralized labor markets.  For
example, we expect it will affect  the junior 
market for economists.  This labor market  typically consists of
short interviews followed by on-campus visits.  Physical constraints
typically limit both the number of short interviews and on-campus
visits a candidate is able to accept.  With virtual interviews and
virtual ``fly-outs,'' we expect candidates to accept more of both than
they would otherwise.  As a result, we could expect the same
bottleneck in the economics job market as in the NRMP match.

\bibliography{NRMPInterviews.bib}

@Article{GaleShapley:AMM1962,
author = {David Gale and Lloyd Shapley},
journal = {American  Mathematical Monthly},
title = {College Admissions and the Stability of Marriage},
pages = {9--15},
volume = {69},
year = {1962}
}

@article{RothPeranson:AER1999,
Author = {Roth, Alvin E. and Peranson, Elliott},
Title = {The Redesign of the Matching Market for American Physicians: Some Engineering Aspects of Economic Design},
Journal = {American Economic Review},
Volume = {89},
Number = {4},
Year = {1999},
Month = {September},
Pages = {748-780},
}

@misc{EcheniqueGonzalezWilsonYariv:2020,
      title={Top of the Batch: Interviews and the Match}, 
      author={Federico Echenique and Ruy Gonzalez and Alistair Wilson and Leeat Yariv},
      year={2020},
       archivePrefix={arXiv}
}

@misc{Kadam:2015,
title={Interviewing in Matching Markets with Virtual Interviews},
author={Sangram V Kadam},
year={2021},
institution={Charles River Associates}
}

@phdthesis{Beyhaghi:2019,
title={Approximately-optimal Mechanisms in Auction Design, Search Theory, and Matching Markets},
author={Hedyeh Beyhaghi},
year={2019},
school={Cornell University},
note={Chapter 5: Two-Sided Matching with Limited Number of Interviews}
}

@article{ChadeSmith:Econometrica2006,
 author = {Hector Chade and Lones Smith},
 journal = {Econometrica},
 number = {5},
 pages = {1293--1307},
 publisher = {[Wiley, Econometric Society]},
 title = {Simultaneous Search},
 volume = {74},
 year = {2006}
}

@article{Agarwal:AER2015,
Author = {Agarwal, Nikhil},
Title = {An Empirical Model of the Medical Match},
Journal = {American Economic Review},
Volume = {105},
Number = {7},
Year = {2015},
Month = {July},
Pages = {1939-78}
}

@Article{Rees-Jones:GEB2018,
  author={Rees-Jones, Alex},
  title={{Suboptimal behavior in strategy-proof mechanisms: Evidence from the residency match}},
  journal={Games and Economic Behavior},
  year=2018,
  volume={108},
  pages={317-330}
  }

@article{AshlagiKanoriaLeshno:JPE2017,
author = {Ashlagi, Itai and Kanoria, Yash and Leshno, Jacob D.},
title = {Unbalanced Random Matching Markets: The Stark Effect of Competition},
journal = {Journal of Political Economy},
volume = {125},
number = {1},
pages = {69-98},
year = {2017}
}

@article{HitschHortacsuAriely:AER2010,
Author = {Hitsch, Gunter J. and Horta\c{c}su, Ali and Ariely, Dan},
Title = {Matching and Sorting in Online Dating},
Journal = {American Economic Review},
Volume = {100},
Number = {1},
Year = {2010},
Month = {March},
Pages = {130-63}
}

@book{RothSotomayor:1990,
place={Cambridge},
series={Econometric Society Monographs},
title={Two-Sided Matching: A Study in Game-Theoretic Modeling and Analysis}, 
publisher={Cambridge University Press},
author={Roth, Alvin E. and Sotomayor, Marilda A. Oliveira},
year={1990},
collection={Econometric Society Monographs}
}

@article{LeeSchwarz:RAND2017,
author = {Lee, Robin S. and Schwarz, Michael},
year={2017},
title = {Interviewing in two-sided matching markets},
journal = {The RAND Journal of Economics},
volume = {48},
number = {3},
pages = {835-855}
}

@article{Lee:REStud2017,
    author = {Lee, SangMok},
    title = "{Incentive Compatibility of Large Centralized Matching Markets}",
    journal = {The Review of Economic Studies},
    volume = 84,
    number = 1,
    pages = {444-463},
    year = 2017,
    month = 09
    }

@techreport{AliShorrer:2021,
  title={The College Portfolio Problem},
  author={Ali, S Nageeb and Shorrer, Ran I},
  year={2021},
  institution="Penn State",
  type="Working Paper"
}

@article{ MorganWinkelEtAl:JSE2020,
title={The Case for Capping Residency Interviews},
journal={Journal of Surgical Education},
volume={78},
number={3},
pages={755-762},
year={2020},
month={09},
author={Morgan, Helen Kang and  Winkel, Abigail F and Standiford, Taylor and Mu\~{n}oz, Rodrigo and  Strand, Eric A and  Marzano, David A and Ogburn, Tony and Major, Carol A and Cox, Susan and Hammoud, Maya M}
}
\appendixpage
\appendix
\section{Role of Interviews:  Preference Formation and Coordination}
\label{sec:role-interviews-play}

The interviews that precede the NRMP match serve at least two important
purposes. The more obvious one is preference formation:
while agents have a prior sense of their preferences, they update their preferences based on information revealed by the interviews. Less obvious is the coordination of whom to
rank. A pair of agents can only be matched if they
\emph{rank one another}. In a market with hundreds, if not thousands,
of potential partners, even formulating---much less submitting---a
ranking of all possible partners is impractical. Thus agents submit
only a portion of their preferences. Interviews function as a device to
coordinate which portions they submit.

The second role of interviews that we have mentioned above  is
important even in a world where interviews are completely
uninformative.

A natural question is to ask how the analysis of this paper is
affected when we take preference formation into account.  We first
observe that we have drawn conclusions about doctors' welfare (\cref{thm:
more not better}) and stability (\cref{prop: sufficient equal l
  k} and \cref{prop: target interview cap}). These comparisons are
based on exogenous preferences. If preference formation
is part of the model, then preferences are endogenous and there is no
basis for such comparisons. However, we explain below how even here
 the match rate, which is an objective measure, is maximized when
there is balance between the two 
sides' interview capacities.

Represent the preferences of the agents prior to the interviews by the
\textbf{pre-interview
  preferences, \boldmath $\overline P$}$\in \mathcal P$. As with 
\cref{prop: sufficient equal l k,prop: target interview cap}, we assume that $\overline
P$ has common preferences.

Given an interview matching $\nu$, for each agent $i\in D\cup H$, let
{\boldmath $P_i^\nu$}  be $i$'s \textbf{post-interview
  preferences}. These are the preferences
that $i$ forms through the interview process. The only requirement  for $P^\nu_i$
is that $\nu(i)$ is the set of partners that $i$ ranks as 
acceptable. We do not make any other assumptions
about how $P^\nu$ relates to $\overline P$. In fact, we leave the exact
orderings unspecified as this is unnecessary to draw
conclusions about the match rate.

Since the pre-interview preferences are common, there is a unique
interview matching, $\nu$, that is stable with respect to these preferences. 
Given $\nu$,  unlike the two-phase process in \cref{sec: market}, the
next step is to use the \emph{post-interview preferences},  $P^\nu$
(as opposed to $(\overline P_i|_{\nu(i)})_{i\in D\cup H}$)  as the 
input to the doctor-proposing DA algorithm. The output is what we call
the \textbf{\boldmath $(l,k)$-matching with updating}.

If the interviews either lead to a
  matching ($l=k=1$) or do not impose any constraints ($l, k \geq
  \min\{|D|,|H|\}$), then a total of
  $\min\{|D|,|H|\}$ pairs form. In the proof of \cref{prop: max
    match rate updating} below, we demonstrate that when $l=k$, the
  $(l,k)$-matching with updating matches as many pairs as this ``first
  best'' benchmark. Increasing the gap between $l$ and $k$ causes the 
  $(l,k)$-matching with updating to fall short of this benchmark
  regardless of what the updated  preferences are. This is a
  straightforward application of the pigeon hole principle.

\begin{proposition}\label{prop: max match rate updating}
Fix the hospitals' interview capacity at $l$ and consider $k$ and $k'$
such that 
either $k' < k \leq l$ or  $l \leq k < k'$. The $(l,k')$-\assignment
with updating has a weakly lower match rate than the
$(l,k)$-\assignment with updating.
\end{proposition}
\begin{proof}
  Let $\{d_t\}_{t=1}^{|D|}$ and
  $\{h_t\}_{t=1}^{|H|}$ be enumerations of $D$ and $H$, respectively,
  such that every hospital prefers $d_t$ to $d_{t+1}$ and every doctor 
  prefers $h_t$ to $h_{t+1}$ under $\overline P$.

  Let $m = \min\left\{\left\lfloor \frac{|H|}{k}\right\rfloor,
    \left\lfloor \frac{|D|}{l}\right\rfloor\right\}$.
  For each $n = 1, \dots, m$, let
  \begin{align*}
    D_n &= \{d_t: (n-1)l < t \leq nl\}\\
    &\text{and}\\
    H_n &= \{h_t: (n-1)k < t \leq nk\}.
  \end{align*}
  By definition of $m$, either $|D| < (m+1)l$ or $|H| < (m+1)k$. Let
  $t_D =  \min\{|D|, (m+1)l\}$,  $t_H =  \min\{|H|, (m+1)k\}$, 
  \begin{align*}
    D_{m+1} &=\left\{d_t: ml < t \leq t_D \right\},\\
            &\text{and}\\
    H_{m+1} &=\left\{h_t: mk < t \leq t_H \right\}
  \end{align*}
  The unique pairwise stable interview 
matching, $\nu$, is such that for each $n = 1, \dots, m+1$, every
doctor in $D_n$ is interviewed by every hospital in $H_n$---that is, for
each $d\in D_n$, $\nu(d) = H_n$ and for each $h\in H_n, \nu(h) =
D_n$. Moreover, every doctor or hospital not in one of these sets has
no interviews---that is, if there is $d_t$ such that  $t> t_D$,
then $\nu(d_t) = \emptyset$ and if there is $h_t$ such that $t> t_H$,
then $\nu(h_t) = \emptyset$.

Now we consider the $(l,k)$-matching with updating, $\mu$. Having received no
interviews,  any  doctor or hospital with index 
higher than $t_D$ or $t_H$, respectively, is necessarily unmatched
regardless of the post-interview preferences.
Since  each agent's post-interview preferences only rank those agents in one's
interview matching as acceptable, for each $n = 1,\dots, m+1$, we have that
\begin{enumerate}
\item for each $d\in D_n$, $\mu(d) \in H_n\cup\{d\}$ and
\item for each $h\in H_n$, $\mu(h) \in D_n\cup\{h\}$.
\end{enumerate}

If $l < k$, then for each $n=1,\dots,m$,  $|D_n| = l < k = |H_n|$. So
exactly $k-l$ hospitals in $H_n$ are unmatched. There are $m(k-l)$
such hospitals. If there are fewer hospitals in $H_{m+1}$ than
doctors in $D_{m+1}$, then $|H_{m+1}| - |D_{m+1}|$ additional
hospitals are unmatched. That is, the number of unmatched hospitals
among those in $H_{m+1}$ is  $\max\{0,|H_{m+1}| - |D_{m+1}|\} =
\max\{0, (t_H - mk) - (t_D - m_l)\} = \max\{0, (t_H - t_D)-
m(k-l)\}$. So the total number of unmatched hospitals is
\[
  m(k-l) +\max\{0, (t_H - t_D)-m(k-l)\},
\] which is weakly increasing in $k$.

If $k < l$, then every hospital $h_t$ such that $t\leq mk$ is matched
by $\mu$.  As when $ l < k$, $\max\{0, (t_H - t_D)-
m(k-l)\}$ hospitals among those in $H_{m+1}$ are unmatched by
$\mu$. Every hospital $h_t$ such that $t > t_H$ is also
unmatched. Thus, the number of hospitals that $\mu$ leaves unmatched is $\max\{0, |H| - ( mk + |D| -
ml)\} = \max\{0, |H| -|D| +m(l-k)\}$, which is weakly decreasing in $k$.
\end{proof}
\cref{prop: max match rate updating}  shows that the bottleneck
of imbalanced interview capacities occurs at the interview stage. In this
sense, the preference formation role of interviews is orthogonal to our  main point
 regarding the harm such imbalance causes.

\section{Doctor-optimal Interview Matching}
\label{sec:doct-optim-interv}
The only difference between our model and that of
\citet{EcheniqueGonzalezWilsonYariv:2020} is that we suppose that the
interview matching is 
hospital-optimal rather than 
doctor-optimal. Their modeling choice is natural
for the question they ask as it gives each doctor her \emph{best}
stable set
of interviews. Thus, their result that most doctors match with
hospitals they rank highly can only be
stronger for other interview matchings. For our analysis, the
doctor-optimal interview matching does not have this natural
appeal. To the contrary, hospital-proposing DA is a
reasonable approximation of the interview matching process.
Nonetheless, our results are not driven by this choice. The only proof
that relies on this choice  is that of  \cref{thm: more not better}. In this section, we
show that the result holds even  for the doctor-optimal
interview matching followed by the doctor-optimal final matching. In
what follows, we use the same terminology and notation as before,
with the understanding that the interview matching is doctor-optimal.

As in the statement of the theorem, suppose that for each $d\in D$,
$\k_d \leq \k'_d$. We show below that the \cref{nobetter,oldisnew} hold
even with the change from the hospital-optimal interview matching to
the doctor-optimal interview matching. The key is to
establish that, in the interview phase, 
if a doctor $d$ is rejected by a hospital $h$ under capacities $\k$,
then $h$ rejects her under $\k'$ as well.
Given capacities $\bar \k$, let $A_d(m;\bar{\k})$  be the hospitals that
$d$ proposes to and  $R_d(m;\bar{\k})$ be  the 
hospitals that reject doctor $d$ by the end of round $m$ of the
interview phase. We show that these sets are monotonic in $\overline \k$.

\begin{claim}\label{claim: doctor-doctor}
For any positive integer $m$,
\begin{align*}
R_d(m;\k) &\subseteq R_d(m;\k')\\
A_d(m;\k) &\subseteq A_d(m;\k')
\end{align*}

\end{claim}

\begin{proof}
We proceed by induction on $m$, the base case being
$m=1$. In the first round of the interview phase, each $d\in D$ proposes to her 
favorite hospitals up to her interview capacity.  Since every doctor 
proposes to at least as many hospitals under $\k'$ as under $\k$, every
hospital receives at least as many proposals under $\k'$ as under $\k$.  Therefore, if a doctor $d$
is rejected by a hospital $h$  in the first round of the interview phase under $\k$, she is
also rejected  by $h$ in the
first round under $\k'$.
Now consider a round $m>1$ of the interview phase and suppose for
each doctor $d$ that $R_d(m-1;\k)\subseteq R_d(m-1;\k')$ and
$A_d(m-1;\k)\subseteq A_d(m-1;\k')$.  In round $m$, each $d\in D$ proposes
to her favorite hospitals that have
not yet rejected her  up to her interviewing capacity.  Under $\k$, $d$ proposes to her $\k_d$ favorite
hospitals in $H \setminus R_d(m-1;\k)$.  Under $\k'$, 
she proposes to her $\k'_d$ favorite hospitals in   $H \setminus
R_d(m-1;\k')$.  By the inductive hypothesis, $H \setminus
R_d(m-1;\k')\subseteq H \setminus R_d(m-1;\k)$.  Therefore, if $d$
proposes to $h$ under $\kappa$, either she proposes to $h$  under
$\kappa'$  as well (she
is choosing more hospitals from a smaller set of options) or she
has already proposed to  and has been rejected by $h$ under $\k'$.  In either case, if
$h\in A_d(m;\k)$, then $h\in A_d(m;\k')$.  Since each hospital $h$ 
receives more proposals but its capacity does not change, if
$h$ rejects doctor $d$ under $\k$, she also rejects doctor $d$ when
choosing from a larger set of doctors who have proposed to it under
$\k'$.  Therefore, if $h\in 
R_d(m;\k)$, then $h\in R_d(m;\k')$. 
\end{proof}

We now explain how \cref{claim: doctor-doctor} implies that
\cref{nobetter,oldisnew} hold even when we switch to the
doctor-optimal interview matching. 
Let $\nu$ and $\mu$ be the interview and final matchings respectively,
under $(\i,\k)$.  Similarly, let $\nu'$ and $\mu'$ be the interview 
and final matchings under $(\i,\k')$.

\cref{nobetter} says that for each $d\in D$, if $h\in
\nu'(d) \setminus \nu(d)$, then $\mu(d) \mathrel{P_d} h$.  Given
$\bar \k$, let
$R_d(\bar{\k})$ denote the set of hospitals that reject $d$ in
any round of the interview phase under capacities $\bar{\k}$.
By \cref{claim: doctor-doctor}, 
$R_d(\k) \subseteq R_d(\k')$.
Under $\k$, $\nu(d)$ consists of $d$'s  $\k_d$ most preferred
hospitals in $H\setminus R_d(\k)$.  That is, the $\k_d$ highest-ranked hospitals
that did not reject her.  Under $\k'$, $\nu'(d)$ comprises $d$'s $\k'_d$
most preferred  hospitals in $H\setminus R_d(\k')$.  As $H\setminus
R_d(\k') \subseteq H\setminus R_d(\k)$, if $h'\in \nu'(d)\setminus
\nu(d)$, then for every $h\in \nu(d)$, $h \mathrel{P}_d h'$.  In
words, since $d$ is interviewed by $h$ under $\k$, she was not rejected by $h$
under $\k$.  As
more hospitals rejected $d$ under $\kappa'$ than under $\k$, $h$ does not
reject $d$ under $\k$.  Therefore, $d$ could have proposed to  $h$
under $\k$, but she chose not to.  Therefore, by revealed preference,
she prefers all  hospitals in $\nu(d)$ to any
of her ``new'' interviews under $\k'$ (those in $\nu'(h)\setminus\nu(h)$).

Lemma 3 said that if $d\in v'(h)$, $d' \in v(h)$, and $d' \mathrel{P_h}
d$, then $d' \in v'(h)$.  By \cref{claim: doctor-doctor}, $d'$ proposes
to at least as many hospitals in the interview phase under $\k'$ as
under $\kappa$.  Since $d'$ 
proposes to $h$ under $\kappa$, she also proposes to $h$ under
$\kappa'$.  Each hospital $h$ accepts its $\i_h$ favorite applicants, so if it
accepts $d$, it must also accept $d'$. 

Since \cref{nobetter,oldisnew} hold, the remainder of the
proof follows exactly as in \cref{sec: main}.

\section{Choice Functions for Interview Phase}
\label{sec:choice-funct-interv}

In the interview phase, we compute a many-to-many matching. However,  each doctor and each hospital only
ultimately matches to at most one other partner, and each has strict
preferences 
over partners.
This necessitates the definition of a choice
function over sets of partners. Consistent with the assumption of
non-strategic behavior with complete information, we focus
on acceptant choice functions that are responsive to preferences over
partners and constrained by interview capacity. That is, given $P\in
\mathcal P$,
\begin{itemize}
\item From the set $H'\subseteq H$, each $d\in D$ chooses the $\k_d$
  best elements of $H'$ according to $P_d$:
  \[
    C_d(H') = \left\{
      \begin{array}{ll}
        \{h\in H': h\mathrel P_d d\}& \text{if }|\{h\in H': h\mathrel
                                      P_d d\}|\leq \k_d\text{ and}\\ \\
        B\subseteq \{h\in H': h\mathrel P_d d\}& \text{such that } |B|
                                                 = \k_d \text{ and for
                                                 each }h\in
                                                 B\\
                                    &\text{and each }h'\in H'\setminus
                                      B, h\mathrel P_d h'\text{ otherwise.}
      \end{array}
    \right.
  \]
\item From the set $D'\subseteq D$, each $h\in H$ chooses the $\k_h$
  best elements of $D'$ according to $P_h$:
  \[
    C_h(D') = \left\{
      \begin{array}{ll}
        \{d\in D': d\mathrel P_h h\}& \text{if }|\{d\in D': d\mathrel
                                      P_h h\}|\leq \i_h\text{ and}\\ \\
        B\subseteq \{d\in D': d\mathrel P_h h\}& \text{such that } |B|
                                                 = \i_h \text{ and for
                                                 each }d\in
                                                 B\\
                                    &\text{and each }d'\in D'\setminus
                                      B, d\mathrel P_h d'\text{ otherwise.}
      \end{array}
    \right.
  \]
\end{itemize}

\section{Epilogue: What Actually Happened}
\label{sec:2021-match}

The 2021 ``Match Day''---the day the NRMP announces the results of the
match---was on March 19. In the words of the NRMP President and CEO,
Donna L. Lamb, 
\begin{quote}
  The application and recruitment cycle was upended as a result of the
  pandemic, yet the results of the Match continue to demonstrate
  strong and consistent outcomes for
  participants.\footnote{\href{https://www.nrmp.org/2021-press-release-delivers-strong-residency-match-2/}{NRMP Press
    Release} dated March 19, 2021.}
\end{quote}
Indeed, contrary to our results, there was a 2.6 percent increase in
PGY-1 positions filled.

In this section, we argue that there is reason to be skeptical  about
the claim that the 2021 match is ``strong  and consistent.'' In
particular, we contend that focusing on match rates leads one to miss
an effect that is analogous to our results.

We have assumed that hospitals na\"ively extend 
interview invitations to their most preferred doctors first and that
these trickle down to less preferred doctors. However, hospitals
are not na\"ive in practice and use heuristics.\footnote{Optimal
  strategies depend critically on the probability of matching with a
  doctor conditional on interviewing her
  \citep{ChadeSmith:Econometrica2006}. In our context, these probabilities 
  are dependent not only on other hospitals' preferences and
  strategies  but also on the rest of the hospital's own interview choices. It is
  implausible that hospitals have this information and that they then compute
  the optimal portfolio of doctors to interview.} Most heuristics include an option that has a high
match probability even if it is ranked relatively low.\footnote{Even
  the optimal solution  has this property when being unmatched
is very unattractive relative to matching with lower ranked
candidates 
  \citep{ChadeSmith:Econometrica2006}.} We  call these candidates
``safety'' candidates.
When hospitals offer
interviews  in this way, an increase to $k$ increases
the proportion of hospitals that match  with their safety candidates
rather than the number of unfilled positions. In other words,
hospitals tend to match with lower ranked doctors, the likelihood of
matching for higher tier candidates decreases, and the likelihood of
matching for lower tier candidates increases.

The very limited information that the NRMP has published so far
supports this conjecture. We do not have 
  access to the actual preferences of residency programs, nor their
  submitted rankings. Yet, the reported interview and ranking patterns
  in the
  \href{https://mk0nrmp3oyqui6wqfm.kinstacdn.com/wp-content/uploads/2020/08/2020-PD-Survey.pdf}{2020
    NRMP Program Director Survey} reveal a systematic preference for MD Seniors
  (graduating students from US MD medical schools) over DO Seniors
  (graduating from US DO medical schools)  as well as MD and DO Grads (those who have previously
  graduated from US medical schools but had not matched). Nonetheless,
  the proportion of positions filled by MD Seniors declined slightly while the
  proportions of positions filled by the latter three increased
  slightly.\footnote{See the NRMP Reports for \href{https://mk0nrmp3oyqui6wqfm.kinstacdn.com/wp-content/uploads/2020/06/MM_Results_and-Data_2020-1.pdf}{2020} and \href{https://mk0nrmp3oyqui6wqfm.kinstacdn.com/wp-content/uploads/2021/03/Advance-Data-Tables-2021_Final.pdf}{2021}.} This lends support to our hypothesis that more positions
  were filled by safety candidates. 

  In what follows, we consider a variation of our model and contrast
  the matching patterns generated by the na\"ive behavior and a
  particular heuristic. The goal is to demonstrate the above explanation of how
  increased matches to safety candidates may mask the effects of
  increased interview capacities.

  The  change to the model is that we separate the doctors into
  two tiers---$D_1$ and $D_2$ (so $D_1\cap D_2 = \emptyset$ and $D_1\cup
  D_2 = D$)---such that every hospital prefers every doctor in $D_1$ to
  every doctor in $D_2$.
  We achieve this change by adding $\beta+\gamma + 10$ to $u_h(d)$, as
  specified in \cref{sec:simulations}, for each $d\in D_1$. 
  We have chosen $D_1$ such that $|D_1| =
  \toptier$.

  To model behavior under the heuristic, we use the following choice
  functions for hospitals, as opposed to those defined in
  \cref{sec:choice-funct-interv}.
  From a set
  $D'\subseteq D$, hospital $h\in H$ chooses a safety candidate
from the second tier, if one is available, and fills the remaining slots
by choosing the best options in $D'$. In
order to ensure that the interview matching (the
hospital-proposing Deferred Acceptance)  algorithm terminates, 
 hospitals do not reconsider tentatively accepted interview offers. Thus,
the choice function is parameterized by the number of interview offers
to be made---in each round of the algorithm, the number of doctors a
hospital chooses is its interview capacity minus the number of
tentatively accepted offers in the previous round.
  Thus, when the hospital's preferences over the doctors are
  represented by $P_h$ it chooses $n$ doctors from $D'\subset D$ as
  follows: If $D\subseteq D_1$,
  \[
    C_h(n, D') = \left\{
      \begin{array}{ll}
        \{d\in D': d\mathrel P_h h\}& \text{if }|\{d\in D': d\mathrel
                                      P_h h\}|\leq n \text{ and}\\ \\
        B\subseteq \{d\in D': d\mathrel P_h h\}& \text{such that } |B|
                                                 = n \text{ and for
                                                 each }d\in
                                                 B\\
                                    &\text{and each }d'\in D'\setminus
                                      B, d\mathrel P_h d'\text{ otherwise.}
      \end{array}
    \right.
  \]
  Otherwise, let $d^s$ be chosen uniformly at random from $D'\cap
  D_2$.\footnote{Hospitals randomize their safety choice $d^s$ to
    avoid overlap.} Then,
    \[
    C_h(n, D') = \left\{
      \begin{array}{ll}
        \{d^s\}\cup \{d\in D': d\mathrel P_h h\}& \text{if }|\{d\in D'\setminus\{d^s\}: d\mathrel
                                      P_h h\}|\leq n-1 \text{ and}\\ \\
        \{d^S\}\cup B& \text{where } B\subseteq \{d\in D'\setminus
                      \{d^s\}: d\mathrel P_h h\} \\& \text{is such that } |B|
                                                 = n-1 \text{ and for
                                                 each }d\in
                                                 B\\
                                    &\text{and each }d'\in D'\setminus
                                      (B\cup\{d^s\}), d\mathrel P_h d'\\&\text{otherwise.}
      \end{array}
    \right.
  \]
  We first observe that under the heuristic, the effect that  increasing 
 $k$ has on the match rate all but disappears. In \cref{fig: k v match
   rate naive,fig: k v match rate heuristic}
 we show the effects for the na\"ive behavior and the heuristic
 respectively.\footnote{All parameter values are the same as in \cref{sec:simulations}.}
\begin{figure}
  \centering
  \begin{subfigure}[b]{0.49\textwidth}
    \centering
    \ifimages
    \includegraphics[scale=0.45]{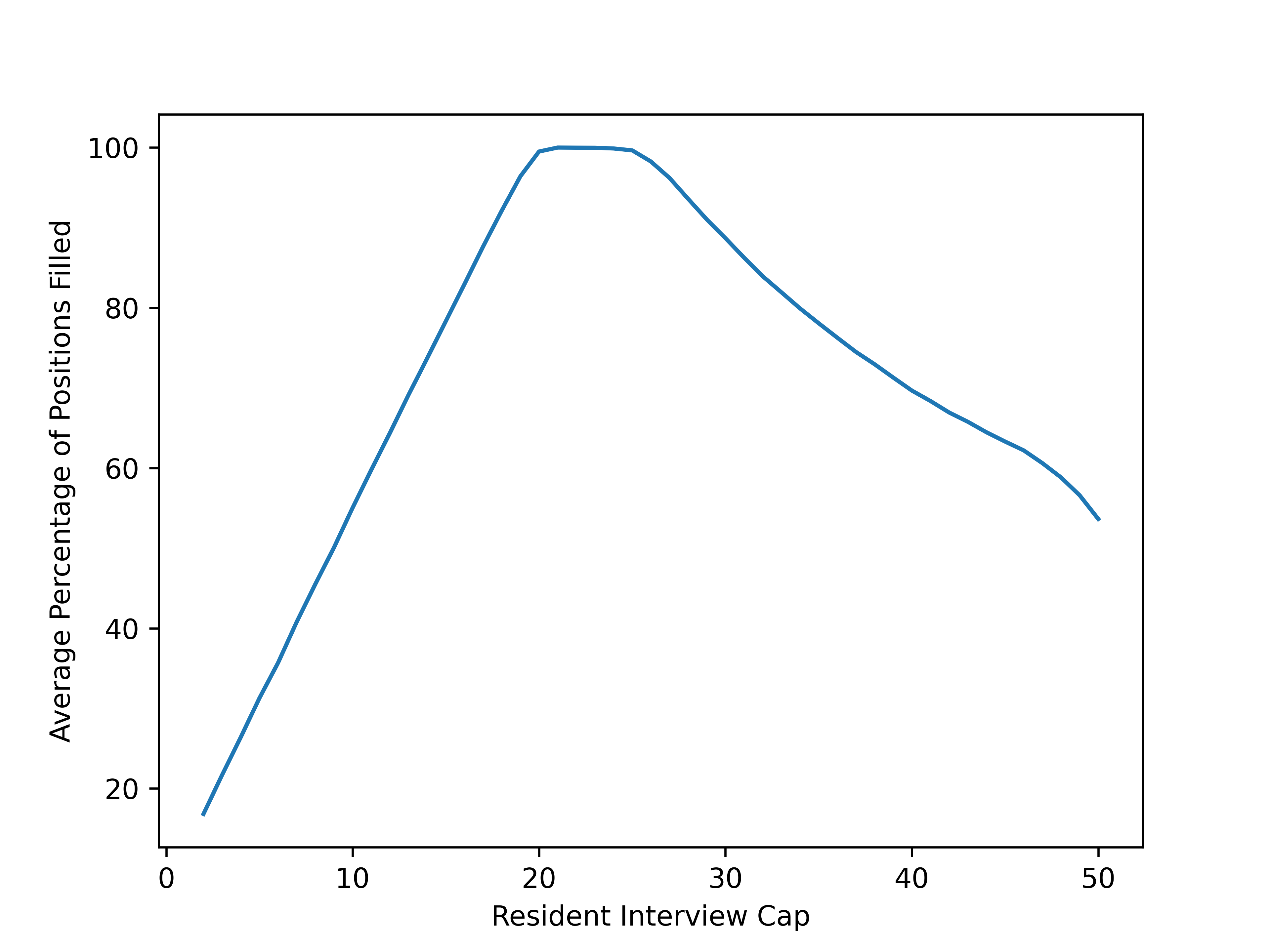}
    \fi 
    \caption{Na\"ive choices }
    \label{fig: k v match rate naive}
  \end{subfigure}
  \hfill
  \begin{subfigure}[b]{0.49\textwidth}
    \centering
    \ifimages
    \includegraphics[scale=0.45]{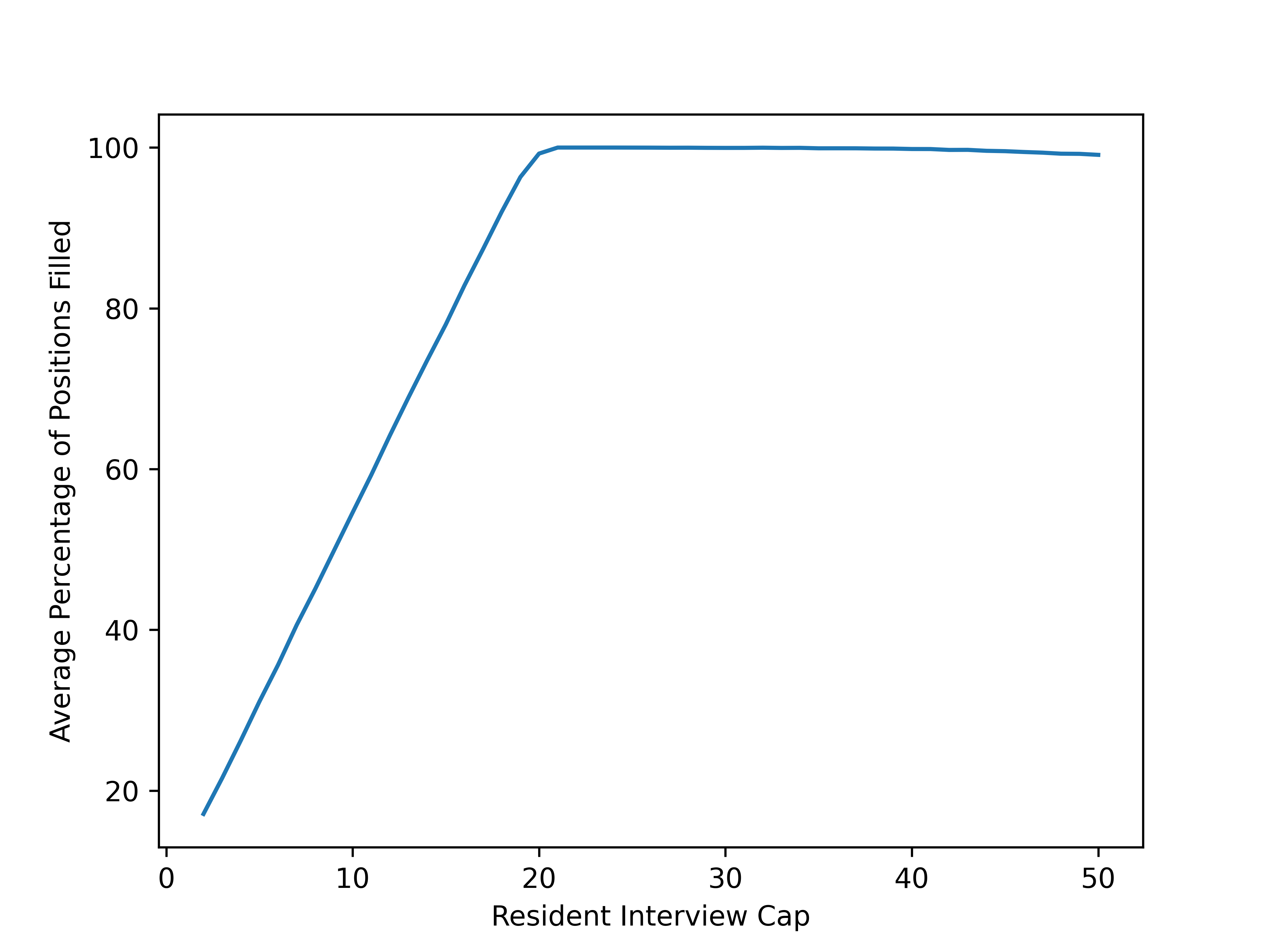}
    \fi
    \caption{Heuristic choices}\label{fig: k v match rate heuristic}
  \end{subfigure}
  \caption{We vary $k$ from $\simlowk$ to $\simhik$ with $l$ fixed at
    $\simhospcap$ and display the corresponding match rate for na\"ive
  behavior by the hospitals, as well as the heuristic where each choice
  includes a safety candidate.}
\end{figure}
The dramatic difference in the match rates is entirely accounted for
by the number of positions that are filled by second tier
candidates. In other words, matching with lower ranked candidates is
the alternative to being unmatched. This is easily seen in the
comparison between \cref{fig: tier match rate naive,fig: teir match rate heuristic}.
\begin{figure}
  \centering
  \begin{subfigure}[b]{0.49\textwidth}
    \centering
    \ifimages
    \includegraphics[scale=0.45]{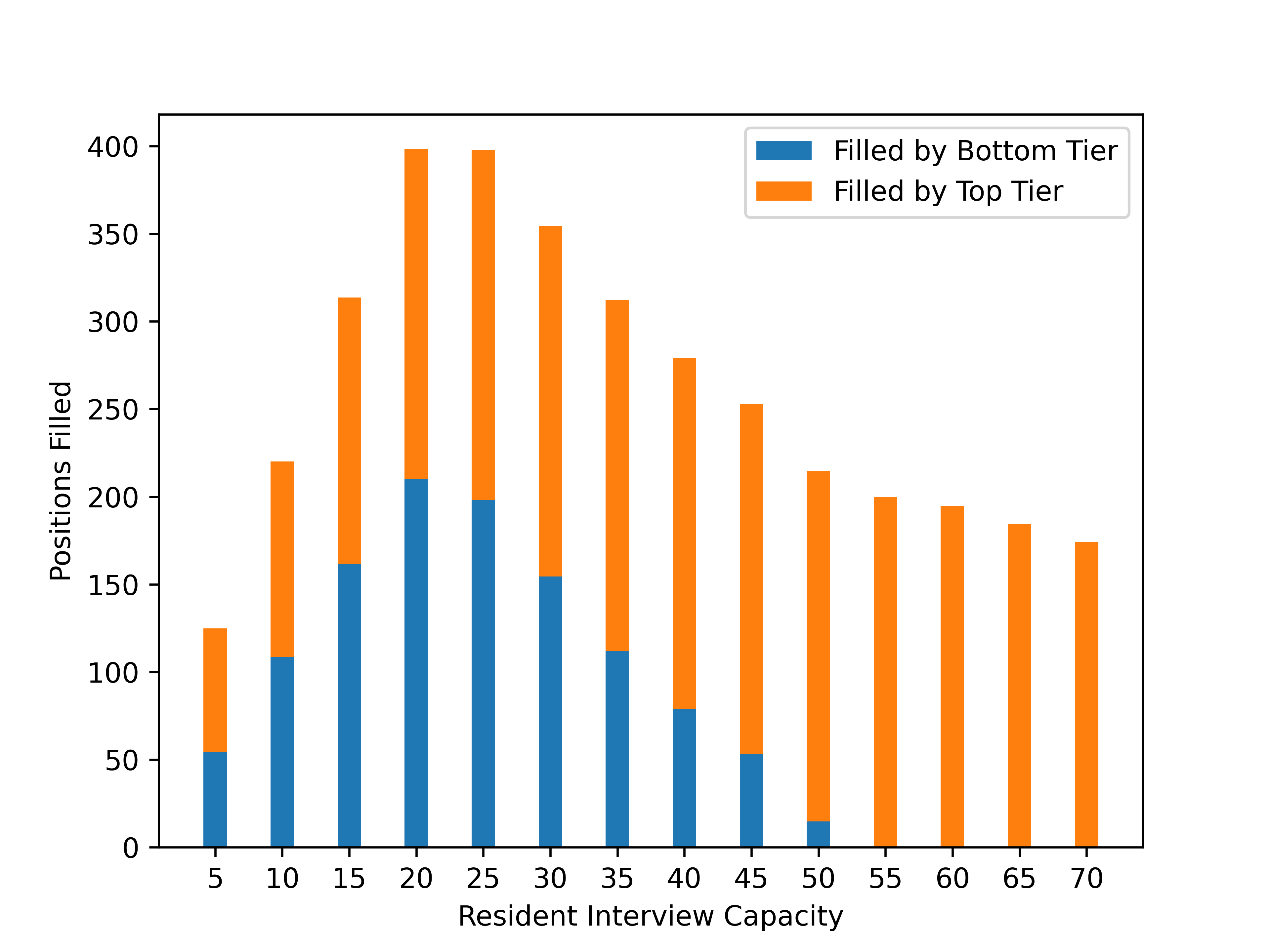}
    \fi 
    \caption{Na\"ive choices }
    \label{fig: tier match rate naive}
  \end{subfigure}
  \hfill
  \begin{subfigure}[b]{0.49\textwidth}
    \centering
    \ifimages
    \includegraphics[scale=0.45]{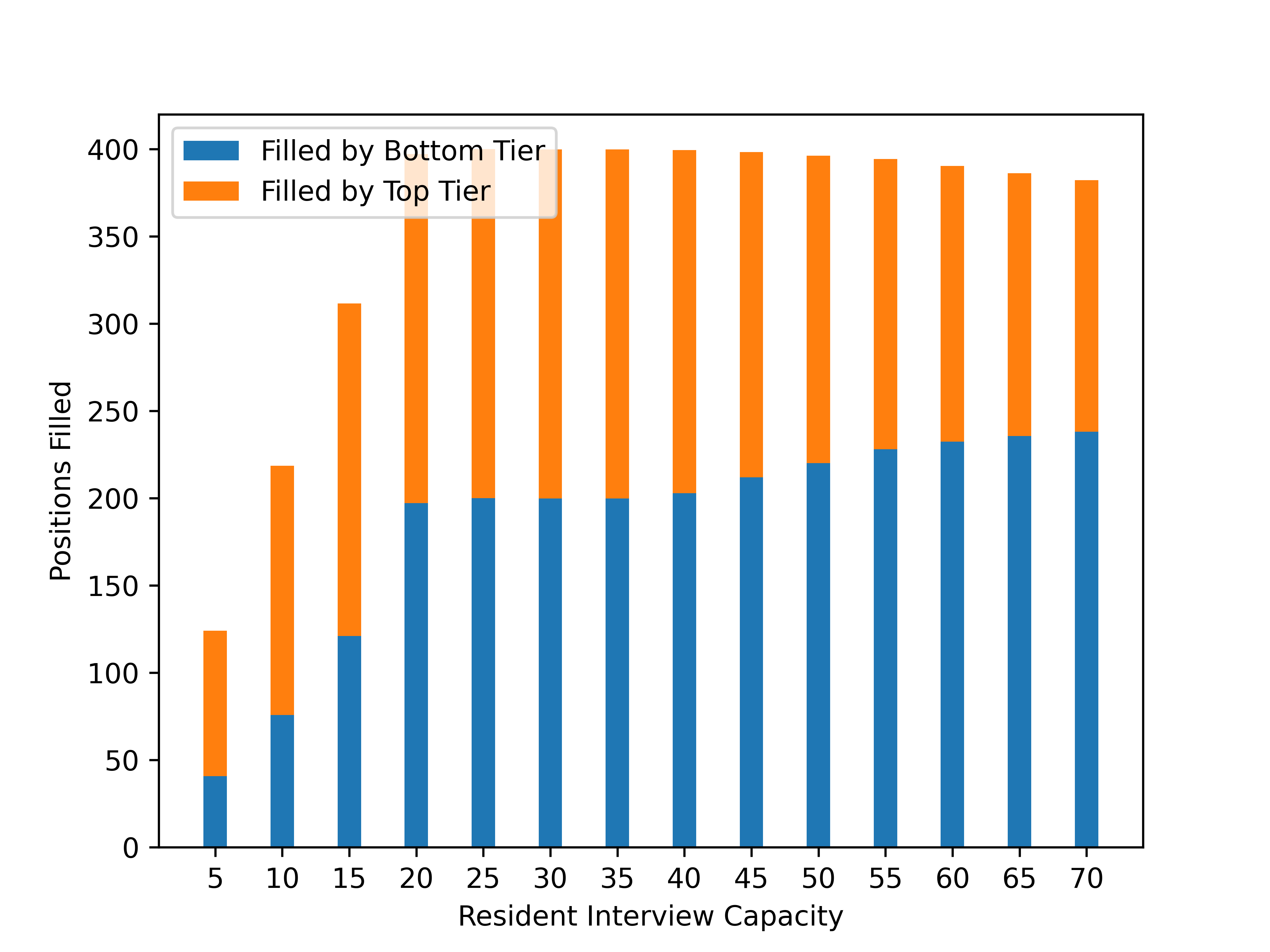}
    \fi
    \caption{Heuristic choices}\label{fig: teir match rate heuristic}
  \end{subfigure}
  \caption{For a range of values of $k$, with $l$ fixed at
    $\simhospcap$, we see that  under na\"ive choices more positions
     are unfilled as $k$ increases. Under the heuristic, the match
     rate does not drop substantially, but positions are
     filled by lower ranked candidates. In either case, the number of
     positions filled by top tier candidates decreases once $k$
     exceeds $l$, which
     is consistent with our theory.}
\end{figure}

Other than the match rate, the remaining patterns that we presented in
\cref{sec:simulations} persist with our modified model under the
heuristic for hospitals' choices.
Under the heuristic, lower ranked doctors who would be unmatched at
lower $k$ gain from the misallocation of interviews among higher ranked
doctors. Nonetheless, more doctors prefer the outcome under a cap of
$k=l$ than prefer the benchmark with no cap as shown in
\cref{fig: tier residents cap preference}. As noted above, hospitals tend to match with
lower ranked doctors and therefore tend to be worse off, as
shown in \cref{fig: tier hospitals cap preference}. The result, as before, is driven by interview
hoarding, which can be seen from the comparison of  interview
distributions with a cap of $k=l$ and with no cap in \cref{fig: tier
  interview dist}.

\begin{figure}
  \centering
  \begin{subfigure}[b]{0.49\textwidth}
    \centering
    \ifimages
    \includegraphics[scale=0.45]{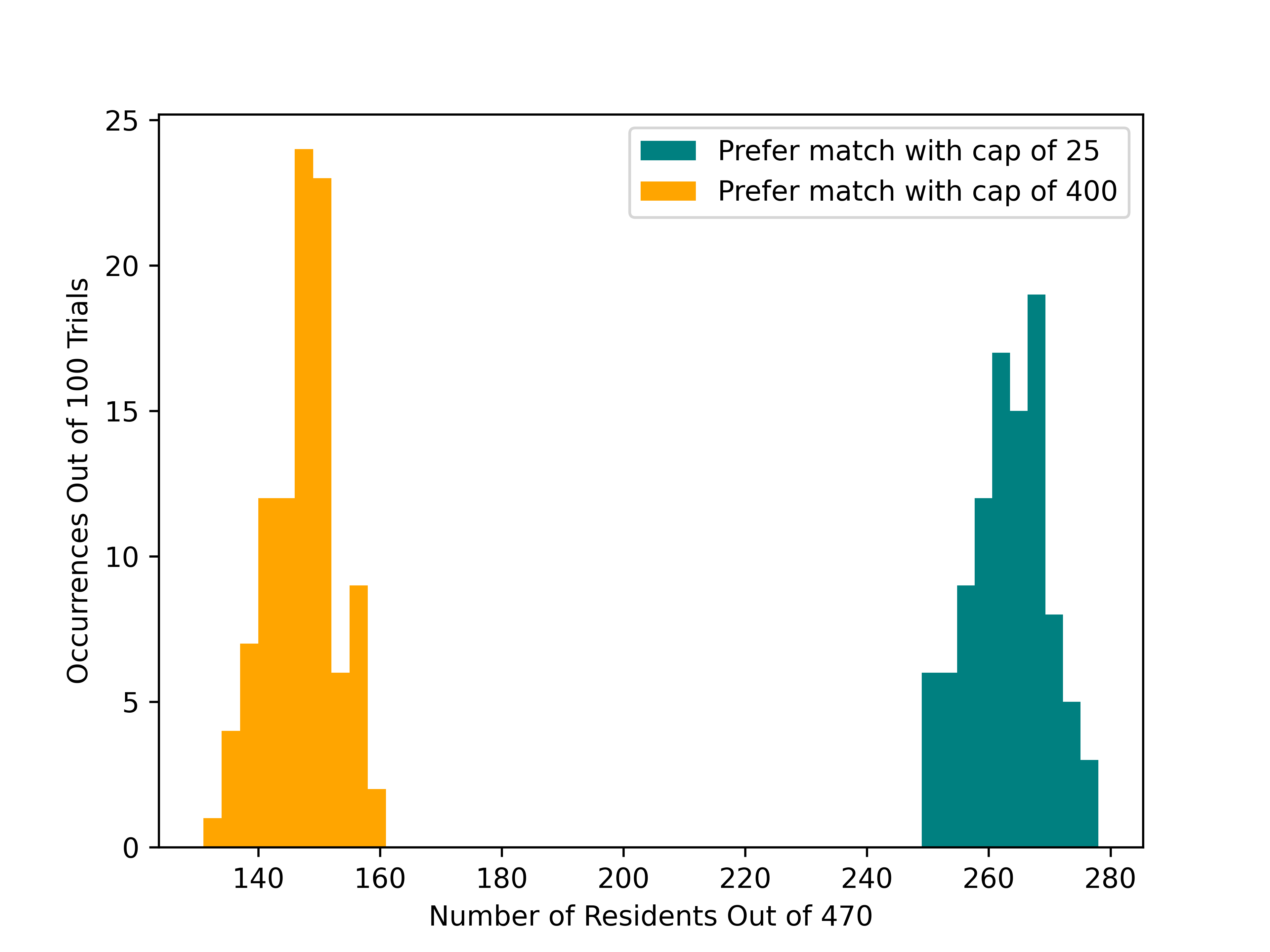}
    \fi
    \caption{Distribution of the number of doctors who prefer their
      match at $k = l$ over being unconstrained and vice versa.}
    \label{fig: tier residents cap preference}
  \end{subfigure}
  \hfill
  \begin{subfigure}[b]{0.49\textwidth}
    \centering
    \ifimages
    \includegraphics[scale=0.45]{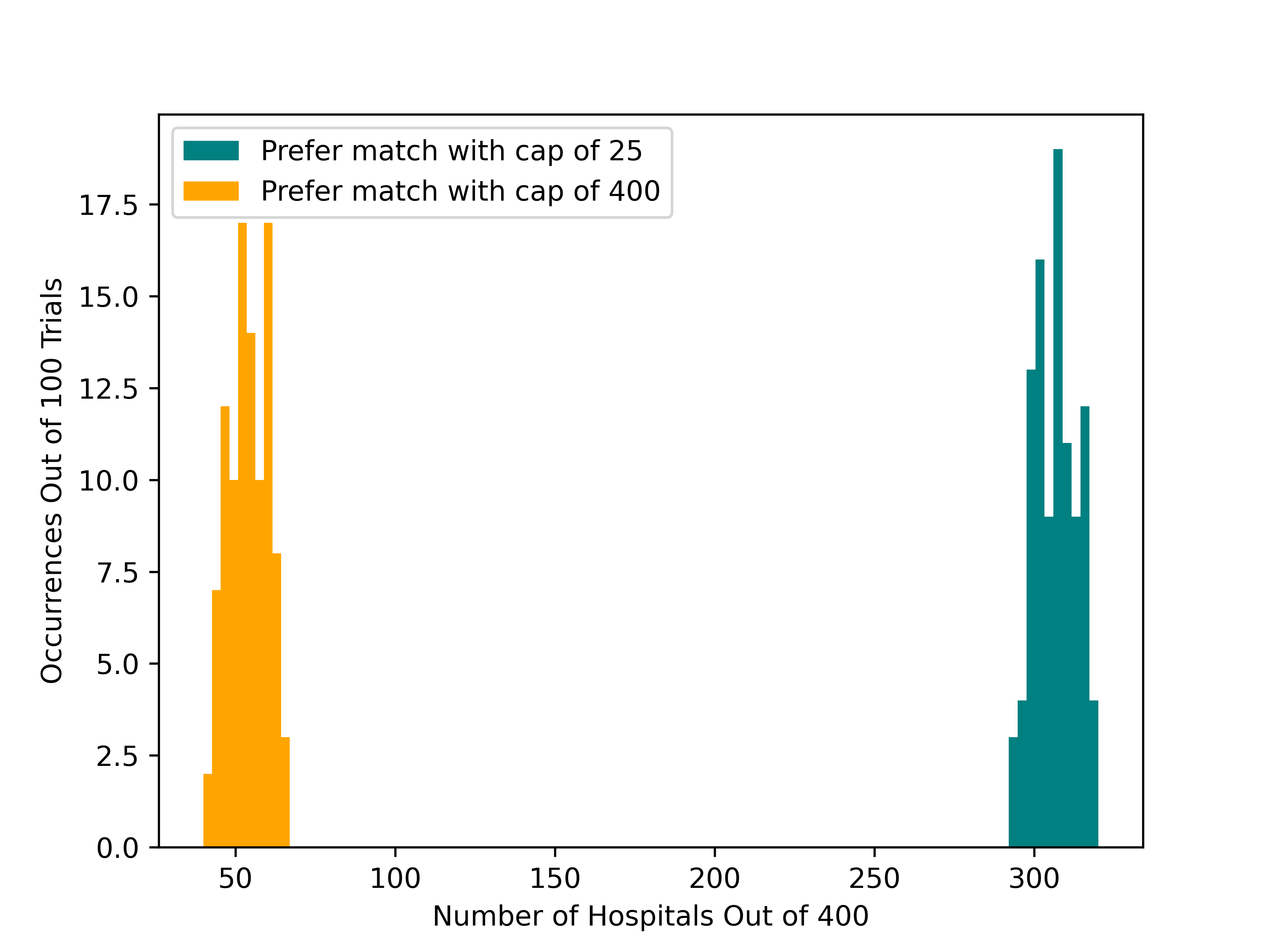}
    \fi
    \caption{Distribution of the number of hospitals that prefer their
      match at $k =l$ over the doctors being unconstrained
      and vice versa. } 
    \label{fig: tier hospitals cap preference}
  \end{subfigure}
  \hfill
  \begin{subfigure}[b]{0.49\textwidth}
    \centering
    \ifimages
    \includegraphics[scale=0.45]{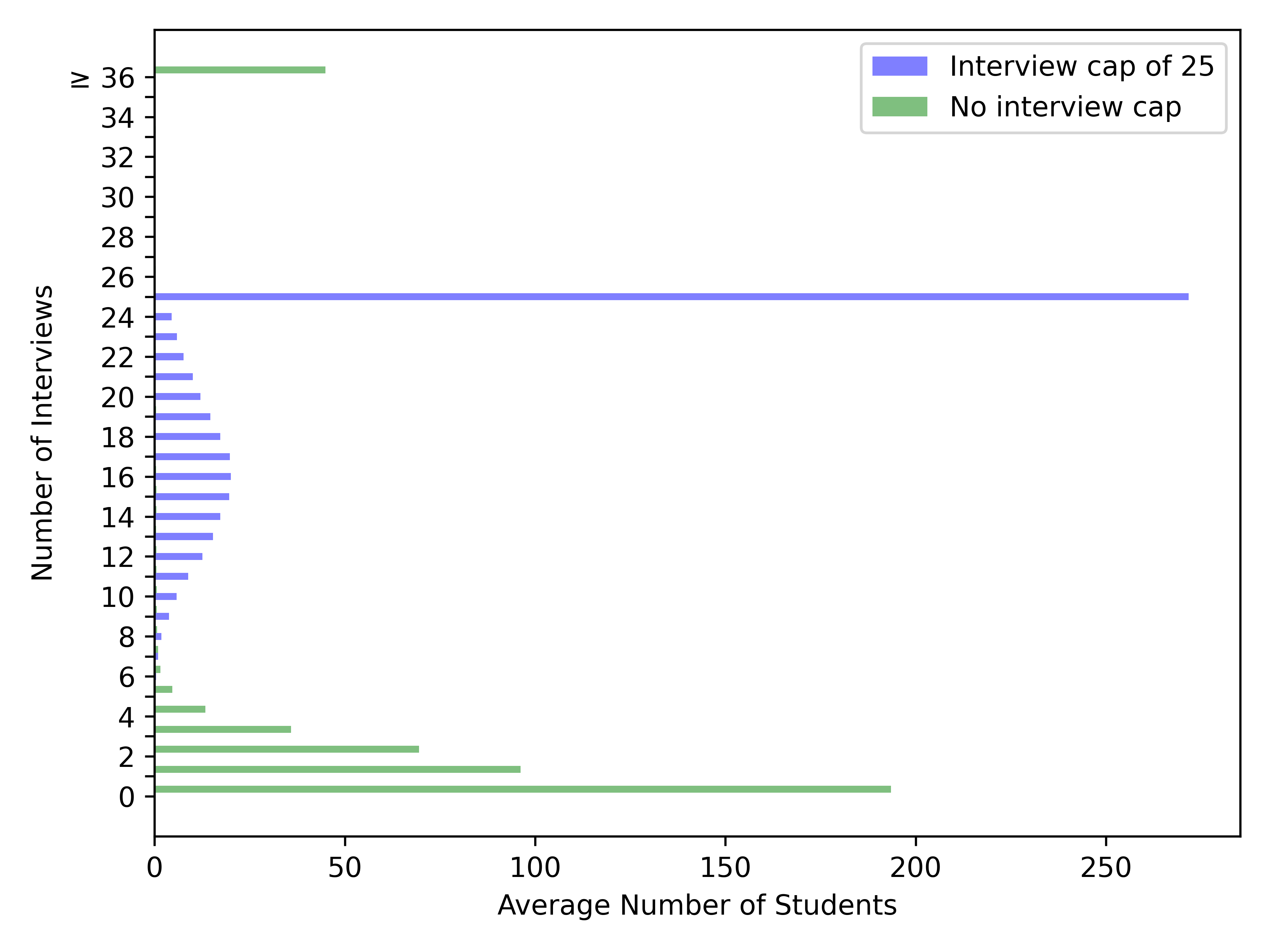}
    \fi
    \caption{Distributions of interviews at $k=l$ and without a
      cap. The uncapped distribution vanishes when the
    number of interviews reach the hundreds.}
    \label{fig: tier interview dist}
  \end{subfigure}
  \caption{Results for the tiered model with heuristic choices for the
    hospitals that parallel our results from \cref{sec:simulations}. }
\end{figure}
 
\section{Robustness of Simulations: Choice of Random Utility Model and Parameters}
\label{apx:sim}

In \cref{sec:simulations}, we have adopted the random utility model of
\citet{AshlagiKanoriaLeshno:JPE2017}. Moreover, we
have presented our results for fixed parameter values: 
$\beta=\simbeta$,  $\gamma=\simgamma$, $l=25$, $|D| = 470$, and $|H| =
400$.  In this appendix, we discuss the robustness of our findings
with regards to these choices.

We start with the random utility model. A very natural alternative is
that of  \citet{Lee:REStud2017}, which is the model that
\citet{EcheniqueGonzalezWilsonYariv:2020}  have adopted. This model
does not include a fit component, so a doctor's utility from being
matched to a hospital is a convex combination of a common (across all
doctors) value and an idiosyncratic (to that doctor) value. A
hospital's utility from being matched to a doctor is similarly
comprised. Thus, each 
hospital $h\in H$ has a common component to its quality, $x_h^C$, and
each doctor has a common component to her quality, $x^C_d$. Aside from
this, for each doctor-hospital pair, $d$ and $h$, $\varepsilon_{dh}$
is the idiosyncratic value that $d$ assigns to $h$ and
$\varepsilon_{hd}$ is the idiosyncratic value that $h$ assigns to
$d$. Then, the utilities that $h$ and $d$ enjoy from being matched to one
another are
\[
  \begin{array}{c}
  u_h(d) = \alpha x_d^C+ (1-\alpha)    \varepsilon_{hd}\\
\text{ and }\\
  u_d(h) = \alpha x_h^C  + (1-\alpha)  \varepsilon_{dh},
  \end{array}
  \]
respectively.

The random variables 
$x_h^C$, $x_d^C$, $\varepsilon_{dh}$, and $\varepsilon_{hd}$ are all
 independently drawn from the uniform distribution over $[0,1]$.  The
 coefficients $\alpha$ and $(1-\alpha)$ are weights on the common and
 idiosyncratic components, respectively.

To make an apples-to-apples comparison between the two models, we
set $\gamma=0$ since the \cite{Lee:REStud2017} model does not have a fit
component. The remaining difference is the distribution of the
idiosyncratic components: the standard logistic distribution in one and
the uniform distribution over $[0,1]$ in the other. Though these
distributions have different supports, we can relate the two models by
considering, for each pair $d,d'\in D$ and each $h\in H$, the degree of
correlation between $u_d(h)$ and $u_{d'}(h)$ as given by the Pearson
correlation coefficient.\footnote{Given the symmetry of both
  models, we could equivalently state this as the correlation between
  $u_h(d)$ and $u_{h'}(d)$ for
  for each pair $h,h'\in H$ and each $d\in D$.} For the \cite{Lee:REStud2017}
model with parameter $\alpha$, it is $\nicefrac{\alpha^2}{\alpha^2 +
  (1-\alpha)^2}$. For the
\cite{AshlagiKanoriaLeshno:JPE2017} model with parameter $\beta$ (and
$\gamma = 0$), it is $\nicefrac{\beta^2}{\beta^2 + (2\pi)^2}$. Thus,
utilities have the same linear correlation in the two models when
\[
  \alpha = \frac{\beta}{\beta+2\pi}.
\]We display this relationship in \cref{fig: alpha beta}. The value of
$\simbeta$ that we have chosen for  $\beta$  in
\cref{sec:simulations} corresponds to an $\alpha$ of
$\simequivalpha$.

\begin{figure}
  \centering
    \ifimages
    \includegraphics[scale=0.4]{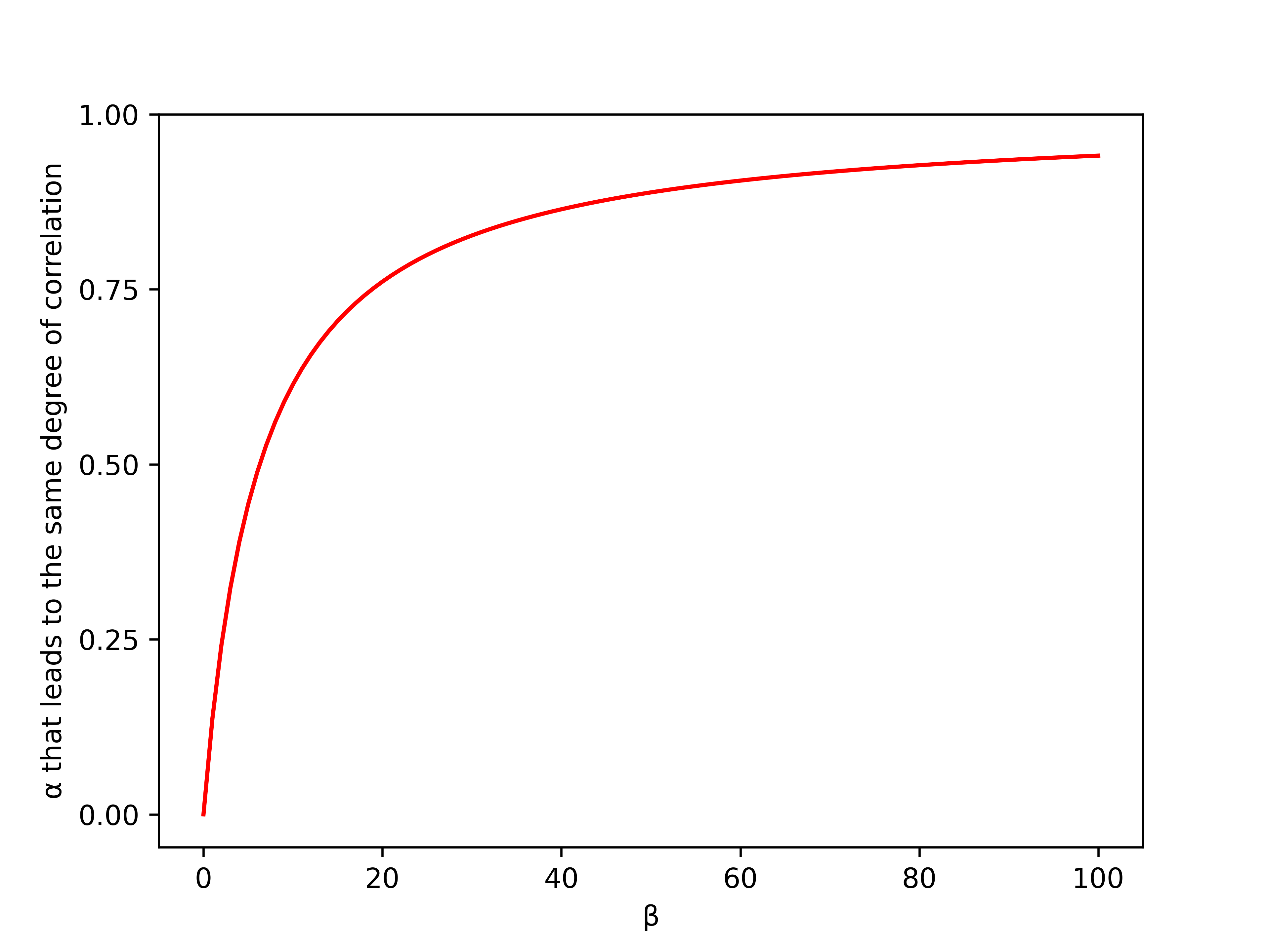}
    \fi
    \caption{Locus of $\alpha$ and $\beta$ pairs that result in the same
    linear correlation between two doctors' (hospitals') utilities from
    the same hospital (doctor) in the models of \citet{Lee:REStud2017} and
    \citet{AshlagiKanoriaLeshno:JPE2017}, respectively}. \label{fig: alpha beta}
\end{figure}

Now consider  $\gamma$. The higher $\gamma$ is, the more aligned
preferences are \emph{across} the two sides of the market. However,
comparing two doctors, when $\gamma > 0$, the further apart they
are in terms of their fit component, the less correlated their
utilities are. Thus, since $\nicefrac{\beta^2}{\beta^2 + (2\pi)^2}$ is
the correlation between the utilities of two doctors with the same fit
component, it is an upper bound on the correlation between the
preferences of any two doctors.
Our chosen value of $\gamma$ is on the same order of
magnitude as $\beta$.

We intend for our simulation results to be suggestive of how the insights
from our analytical results may extend beyond the assumptions that we
make for the sake of tractability. To this end, none of the choices
that we have made are critical in driving the effects we describe in
\cref{sec:simulations}. We consider each of these choices in turn and
make this point by focusing on how the match rate responds to different
values of $k$ as shown in
\cref{fig: k v match rate}.  In other words, we show how this
relationship changes as we vary our model and parameter choices. In
what follows, we vary only one choice at a time, leaving fixed the
other parameters as in \cref{sec:simulations}.

In \cref{fig: match rate alpha}, we consider  various values of
$\alpha$ and use  the  \citet{Lee:REStud2017} model.
As long as there is at least moderate correlation in preferences, we observe  the  effect, albeit with varying magnitude. Indeed,
as the correlation between utilities grows, the effect becomes
stronger, eventually converging to the one described by \cref{prop:
  target interview cap} when $\alpha = 1$.
\begin{figure}
    \centering
    \ifimages
    \includegraphics[scale=0.65]{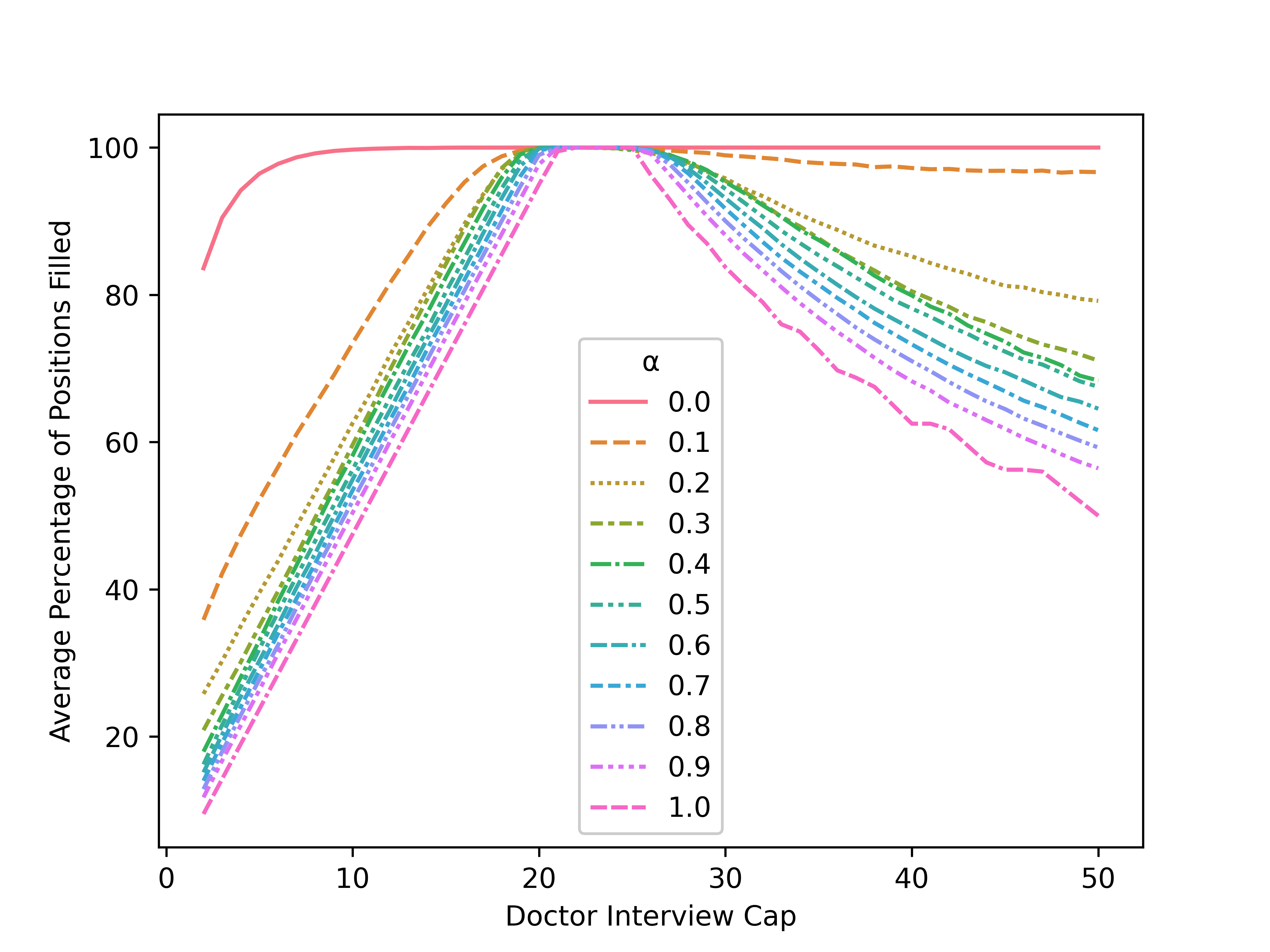}
    \fi
    \caption{The average match rates for each $k$ under the random
      utility model of \citet{Lee:REStud2017} at various values of
      $\alpha$.}
    \label{fig: match rate alpha}
\end{figure}
\begin{figure}
    \centering
    \ifimages
    \includegraphics[scale=0.65]{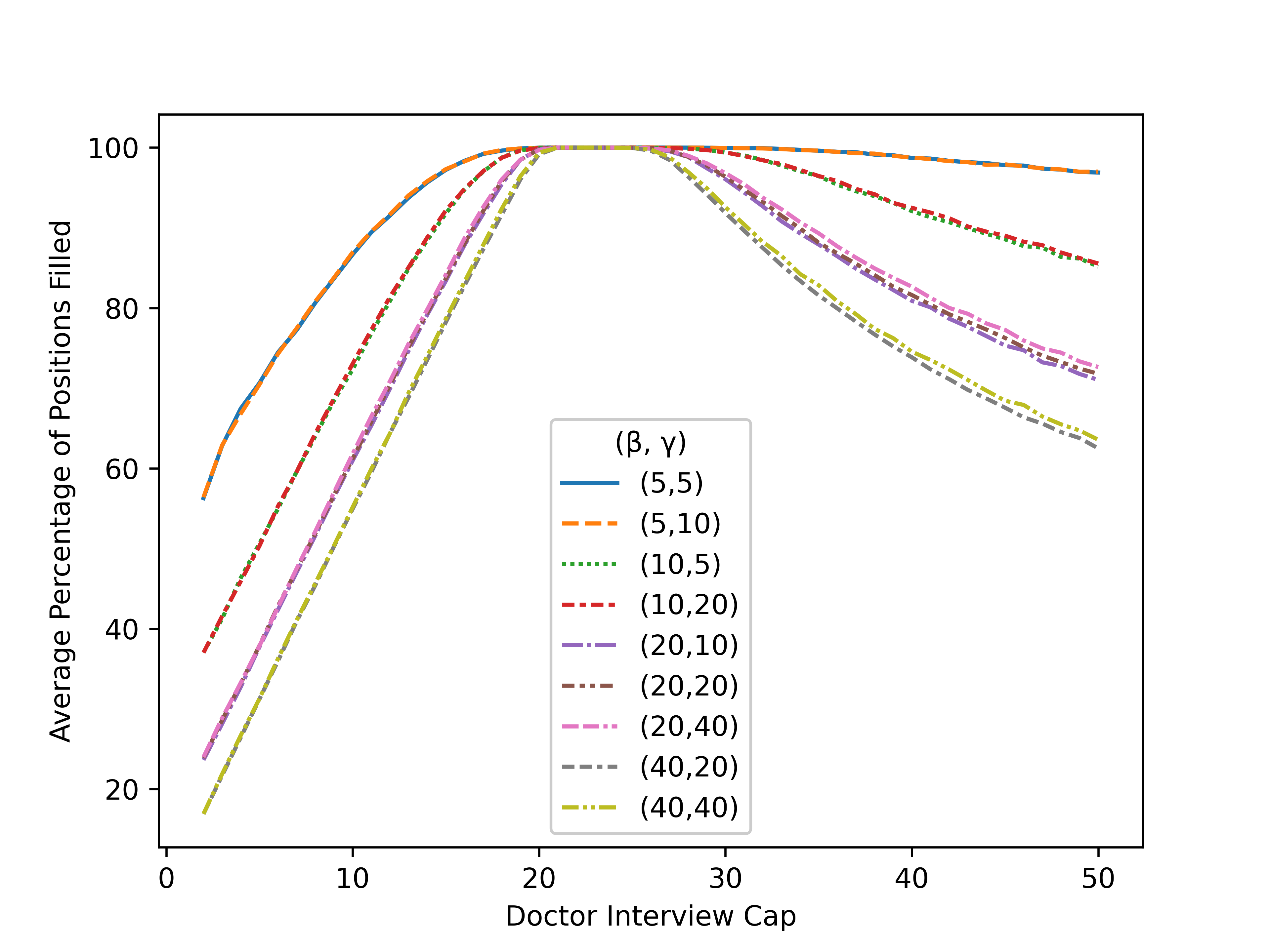}
    \fi
    \caption{The average match rates for each $k$ at various values of
      $\beta$ and  $\gamma$.}
    \label{fig: match rate beta gamma}
  \end{figure}

In  \cref{fig: match rate beta gamma}, we consider various
combinations of $\beta$ and $\gamma$. Again, as $\beta$ grows, the
correlation in preferences increases and the
effect becomes stronger. However, as discussed above, $\gamma >0 $ has
the effect of decreasing the correlation between utilities of agents on the same
side on average. Consistent with this understanding, we see that for a
fixed value of $\beta$, the lower $\gamma$ is, the stronger the effect.

  The next parameter we vary is the number of doctors. Recall that we
 fixed the number of hospitals at $\simnumhosps$ for the simulations
 reported in \cref{sec:simulations}. In \cref{fig: match rate num
   docs}, holding the number of hospitals constant, we  vary the number of
 doctors. The effect
 that $k$ has on the match rate is consistent with the rest of our results.\footnote{For the cases where there are fewer than $\simnumhosps$
 doctors, the  match rate is necessarily lower than  100\%.}

\begin{figure}
    \centering
    \ifimages
    \includegraphics[scale=0.65]{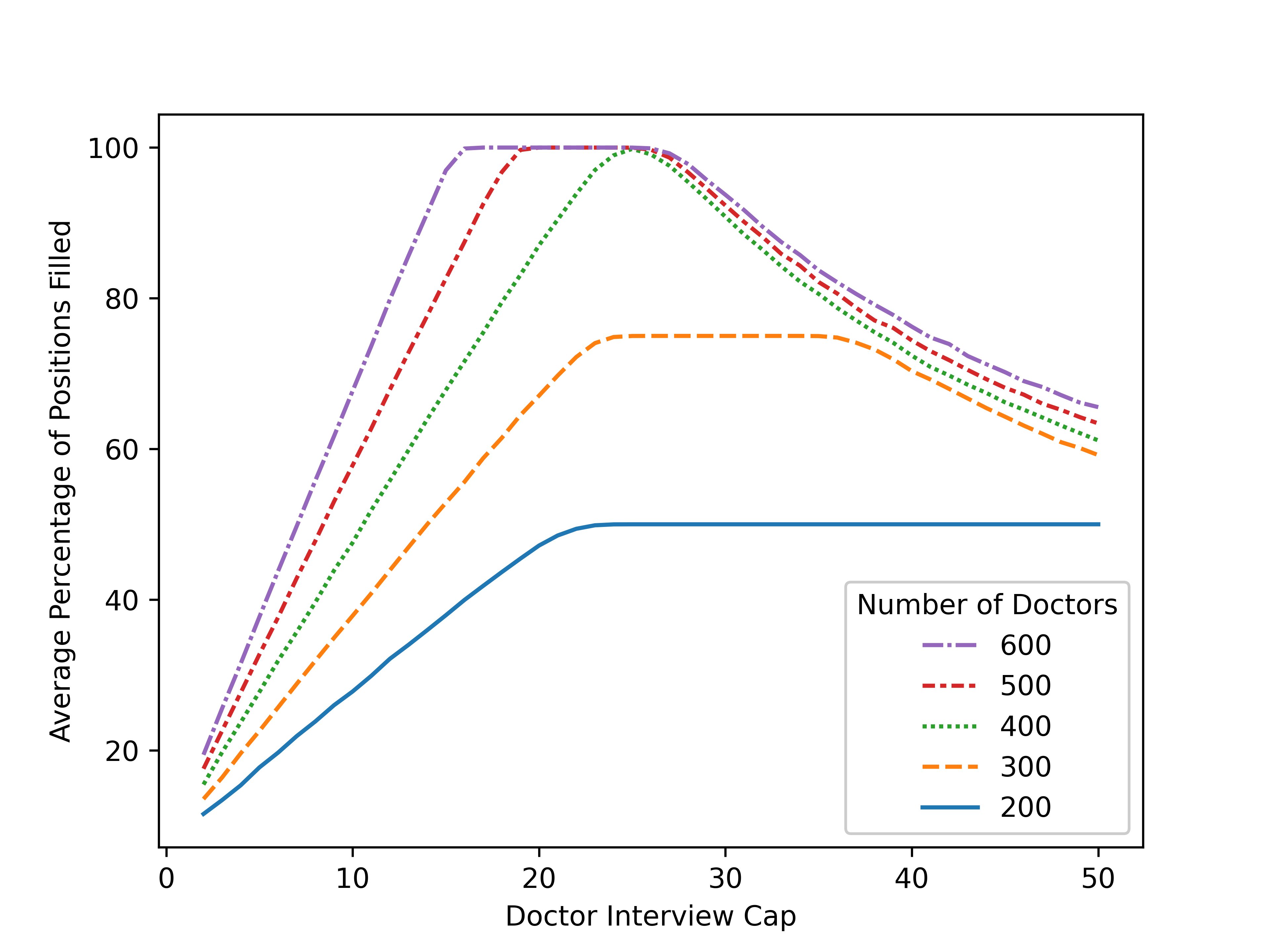}
    \fi
    \caption{The average match rates for each $k$ for different
      numbers of doctors when there are $\simnumhosps$ hospitals.}
    \label{fig: match rate num docs}
  \end{figure}

  Keeping the ratio of hospitals to doctors the same, we next
  proportionally vary the
  size of the market in \cref{fig: match rate market size}. The
  magnitude of the effect increases very slightly with market size,
  but there is no other change.
\begin{figure}
    \centering
    \ifimages
    \includegraphics[scale=0.6]{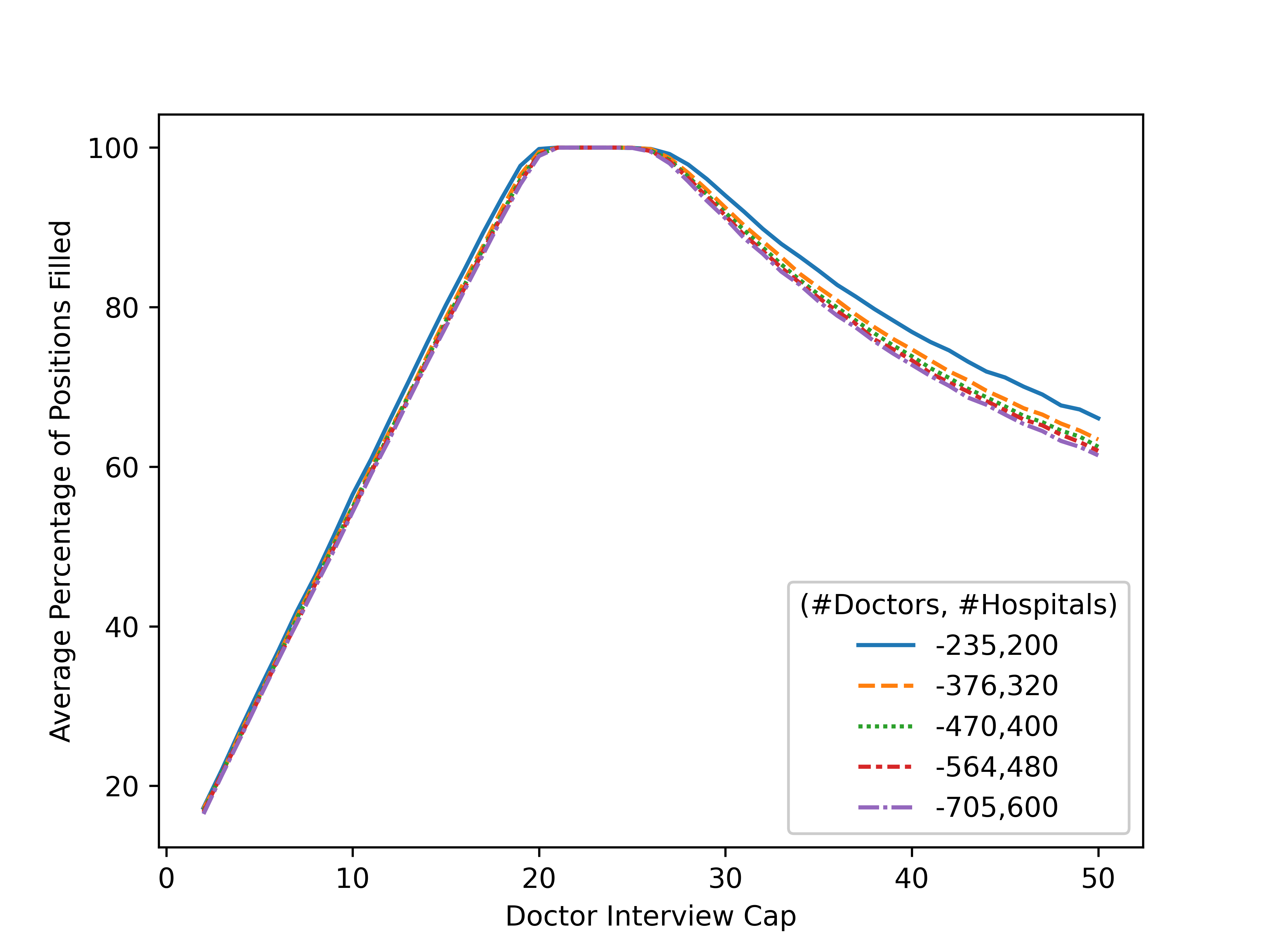}
    \fi
    \caption{The average match rates for each $k$ for different
      market sizes with the ratio of hospitals to doctors fixed at
      $\frac{\simnumhosps}{\simnumdocs}$.}
    \label{fig: match rate market size}
  \end{figure}

  Finally, we consider the value of $l$ that we have fixed at
  $\simhospcap$ for the simulations in \cref{sec:simulations}.
  The lesson of \cref{prop: target interview cap} that  an optimal
  policy is to set $k=l$ holds regardless
  of $l$ (\cref{fig: match rate l}). However, the magnitude diminishes
  with $l$ as see in \cref{fig: k minus l}. Notably, 
the match rate does decrease as
  $k$ becomes larger than $l$.

\begin{figure}
  \centering
  \begin{subfigure}[b]{0.49\textwidth}
    \centering
    \ifimages
    \includegraphics[scale=0.45]{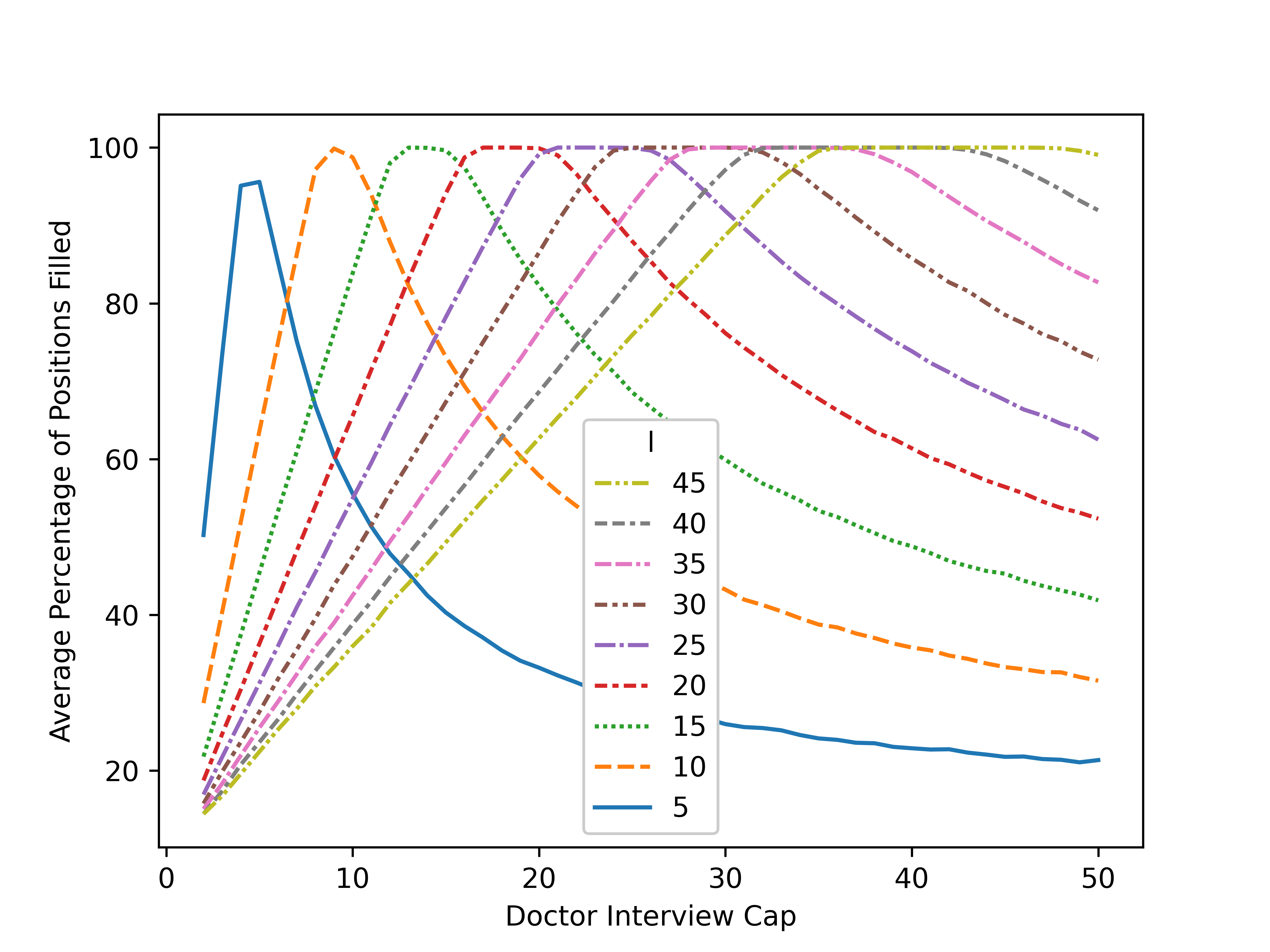}
    \fi
    \caption{The average match rates for each $k$ for different
      values of $l$, the hospitals' interview capacity. \\\hspace{\textwidth}
      } 
    \label{fig: match rate l}
  \end{subfigure}
  \hfill
  \begin{subfigure}[b]{0.49\textwidth}
    \centering
    \ifimages
    \includegraphics[scale=0.45]{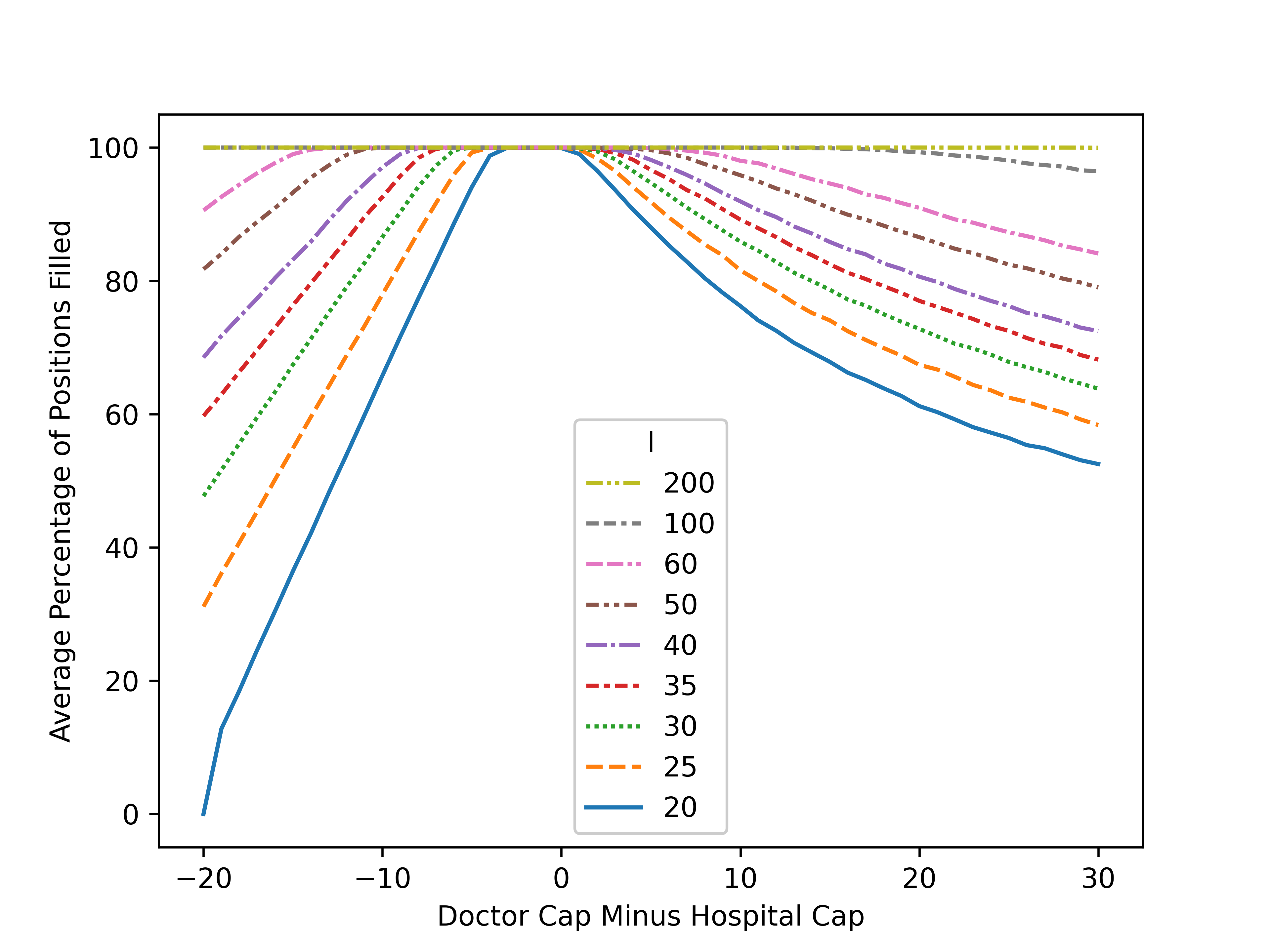}
    \fi
    \caption{The average match rate for different $l$s and $k$s. The
      horizontal axis shows $k-l$. The magnitude of the effect of $k$
      exceeding $l$ diminishes with $l$. } 
    \label{fig: k minus l}
  \end{subfigure}
\caption{The effect of increasing $k$ for different
  values of $l$. }
\end{figure}

 \end{document}